\documentclass[aps,superscriptaddress,nofootinbib,showpacs]{revtex4}

\usepackage{amssymb}
\usepackage{amsmath}
\usepackage{amsthm}
\usepackage{slashed}

\newtheorem{thm}{Theorem}
\newtheorem{cor}{Corollary}
\newtheorem{lem}{Lemma}

\theoremstyle{remark}

\theoremstyle{definition}

\def\drc{D\! \! \! \! /}
\def\dar{\buildrel d\over \longrightarrow }
\def\ar{\longrightarrow }
\newcommand\scri{\mathcal{I}}

\begin{document}

\title{Topological Obstructions To Maximal Slices}

\author{Donald M.\ Witt}
\affiliation{Department of Physics and Astronomy, University of British Columbia,
Vancouver, British Columbia \ V6T 1Z1}
\affiliation{Perimeter Institute for Theoretical Physics, 31 Caroline Street North, Waterloo,
ON, N2L 2Y5, Canada}
\date{\today}

\begin{abstract}
A necessary condition for a globally hyperbolic spacetime ${\mathbb R}\times \Sigma$
to admit a maximal slice is that the Cauchy slice $\Sigma$ admit a metric with nonnegative
scalar curvature, $R\ge 0$. In this paper, the two cases considered are the closed 
spatial manifold and the asymptotically flat spatial manifold. Although most results
here will apply in four or more spacetime dimensions, this work will mainly consider
 4-dimensional spacetimes. For $\Sigma$ closed or asymptotically flat,
all topologies are allowed by the field equations. 
Since all $\Sigma$ occur as Cauchy slices of solutions to the Einstein equations 
and most  $\Sigma$ do not admit metrics with $R\ge 0$, 
it follows that most globally hyperbolic spacetimes never admit a maximal slice, i.e. a 
slice with zero mean extrinsic curvature. In particular, asymptotically flat globally hyperbolic 
spacetimes which admit maximal slices are the exception rather than the rule. 
The reason for this is 
due to topological obstructions to constructing such slices. In the asymptotically flat case, this will be shown by
smooth compactification of the manifold in order to use the results for spatially
closed manifolds.
\end{abstract}

\maketitle

\section{ INTRODUCTION}

Finding necessary and sufficient conditions for the existence of constant mean
curvature (CMC) hypersurfaces in physically reasonable spacetimes has been an
outstanding problem in classical relativity. Such slices are useful  in studying the behavior of singularities, numerical relativity,
and for calculating conserved quantities \cite{e,c}. A particular class of such 
hypersurfaces are maximal slices, namely, slices with zero extrinsic curvature. 
Although, in spatially closed spacetimes, maximal slices are often scarce unless 
 the spacetime is static, it was a common belief until  later work that not only do all 
asymptotically flat globally hyperbolic spacetimes have a maximal slice but, in fact,  have an entire 
foliation by maximal slices. In particular, spacetimes sufficiently near to Minkowski 
spacetime admit maximal slices. Additionally, spacetimes admitting 
time functions obeying certain boundary conditions admit maximal slices; the existence theorems are
due to R. Bartnik \cite{b}. There are additional results for  slices outside black hole 
regions \cite{Bartnik:1990um}. On the other hand, D. Brill, gave examples of some 
spatially closed spacetimes and an asymptotically flat spacetime that admit no maximal 
slices \cite{db,drd}. The asymptotically flat example of Brill was not smooth but only piecewise
smooth and also was for only a limited set of topologies. These restrictions were removed
in \cite{Witt:1986ng}. Hence smooth, asymptotically  flat spacetimes generically exhibit a topological obstruction to maximal slices.

Historically, a less well known problem was whether every 3-manifold admits physically reasonable 
initial data for a spacetime. The solution of this problem gives insight into not only into
the mathematical structure of the classical field equations but also has 
implications in the quantum theory. In quantum gravity, it is believed 
that the topology of spacetime is highly nontrivial on the microscopic 
scale. Since one expects that, at a minimum,  classically allowed topologies 
occur in the quantum theory,  if every 3-manifold admits physically
 reasonable initial data then it is sensible to discuss all spatial 
topologies when studying quantum gravity. 

For the first time ever it was shown in \cite{Witt:1986ng} that all 3-manifolds 
do in fact have physically reasonable initial data. Therefore, there are no obstructions
to the allowed topologies from the Einstein equations. Moreover, this result combined 
with the work on the classification of 3-manifolds which admit metrics with nonnegative 
scalar curvature implies that the generic situation is  that globally hyperbolic spacetimes do not 
admit maximal slices.  In other words, a spacetime which admits a maximal slice is the exception 
rather than the rule. The purpose of the present paper is to provide the     
details not given in the earlier paper  \cite{Witt:1986ng}. 

In  \cite{Witt:1986ng}, the case of physically reasonable matter was treated both  in the case of cosmological models and asymptotically flat spacetimes. The vacuum case for cosmological models was also treated. The asymptotically 
flat vacuum case was treated in n-dimensions in \cite{Isenberg:2002es}.
For a review of the initial data problem see \cite{Rendall:2002yr}.

Physically, all topologies are allowed as spatial topologies of the asymptotically 
flat case and cosmological case. However, it turns out that {\it topological censorship}
implies that  structures with nontrivial first homotopy are censored in the asymptotically flat and asymptotically locally 
 anti-de Sitter case: the structures collapses to form black holes or, more specifically are hidden behind 
black hole horizons \cite{Friedman:1993ty,Schleich:1994tm,Galloway:1999bp,Galloway:1999br}.
In fact the original form of topological censorship, first proven in \cite{Friedman:1993ty}, 
was motivated by the solutions given in \cite{Witt:1986ng}. 

In the present paper, it will be shown that every nonnegative function on a 
closed 3-manifold is the energy density of an initial data set with $J^a=0$. In 
particular, every closed 3-manifold has vacuum initial data. Furthermore, 
there exists initial data on every closed 3-manifold such that it is a CMC hypersurface 
for some constant. Next, it will be shown that every asymptotically flat 
$\Sigma $ has initial data. Moreover, every asymptotically flat $\Sigma $ 
has initial data with  a CMC hypersurface. 
Finally, the initial data sets constructed here can be evolved into globally hyperbolic spacetimes 
of the form ${\mathbb R}\times \Sigma $. Therefore, it makes sense to consider 
all manifolds $\Sigma$ as a possible topology of an initial data set for the Einstein equations. 

A necessary condition for a globally hyperbolic spacetime 
to admit a maximal slice is that $\Sigma $ admit a metric with $R\ge 0$. 
For closed 3-manifolds, M. Gromov and H. Lawson have shown that 
such manifolds comprise a small fraction of all closed 3-manifolds \cite{g}. 
Therefore, only a small fraction of spatially closed, globally hyperbolic spacetimes can admit 
maximal slices. In order to classify the asymptotically flat 3-manifolds $\Sigma $ 
which admit metrics with $R\ge 0$, a compactification theorem is proven 
using the Green's function of the operator $-8D^2 + R$. The compactification 
theorem reduces the classification of asymptotically flat 3-manifolds to   that
of closed ones with $R>0$. Again, it follows that most globally hyperbolic spacetimes 
do not have maximal slices. On the other hand,   
there are no topological obstructions to finding CMC hypersurfaces 
because initial data with $p$ equal to a constant  can be constructed regardless 
of the topology of $\Sigma $. 

The mathematical techniques presented here in this paper provide a set of tools useful to solving other related
problems of interest in gravitational physics. For example, the theorems proven here
apply in higher dimensions and were used to prove existence of certain vacuum solutions in gravity with compact
extra dimensions \cite{Hertog:2003ru}. The authors of \cite{Hertog:2003ru} 
acknowledge this as a private communication with the author of this paper. 

\section{INITIAL DATA SETS} 

A  {\it Cauchy surface} is a spacelike hypersurface such that every non-spacelike curve 
intersects this surface exactly once. A {\it partial Cauchy surface} is a surface that satisfies the 
weaker condition that each non-spacelike curve intersects the surface at most once. 

A spacetime is {\it globally hyperbolic} if and only if it admits a Cauchy slice. Alternately, 
a spacetime $M$, possibly with boundary, is {\it globally hyperbolic} if it is strongly causal 
and the sets $J^+(p,{M})\cap J^-(q,M)$ are compact for all 
$p,q\in M$ \footnote{The timelike future (causal future)
of a set $S$ relative to  $U$, $I^+(S,U)$ ($J^+(S,U)$), is the set of all points that can be reached from 
$S$ by a future directed  
timelike curve (causal curve) in $U$.  The interchange of the past with future in the previous
definition yields $I^-(S,U)$ ($J^-(S,U)$). }.

This definition is a generalization of that of a globally hyperbolic spacetime without boundary 
and is satisfied by asymptotically locally anti-deSitter (ALADS) spacetimes \footnote{ In fact, it is that used in the proof of topological censorship in ALADS spacetimes \cite{Galloway:1999bp}.}.  Also, note 
that the Penrose compactification of an asymptotically flat (AF) globally hyperbolic spacetime 
(which is itself globally hyperbolic by the usual definition) is globally hyperbolic in this general sense.

The {\it domain of outer communications} (DOC) is the portion of a spacetime  ${ M}$ 
which is exterior to event horizons. Precisely ${D} = I^-(\scri^+_0)\cap I^+(\scri^-_0)$ for a connected 
component $\scri_0$ for an AF spacetime and
${D} = I^-(\scri_0)\cap I^+(\scri_0)$ for an ALADS spacetime. Intuitively, the DOC is the subset of ${ M}$ that
is in causal contact with $\scri$. Note that $D$ is the interior of an $(n+1)$-dimensional 
spacetime-with-boundary ${D}' = {D }\cup \scri$ and that $D'$ is itself a globally 
hyperbolic spacetime with boundary.

An {\it event horizon} is the boundary of the  DOC. More specifically, a {\it future event horizon} is 
the boundary of the causal past of a connected component of the boundary at infinity, 
$\scri_0$, $ \dot J^{-}(\scri_0, M')$, a {\it past event horizon}
is the boundary of the causal future of $\scri_0$, $ \dot J^{+}(\scri_0, M')$ and the 
{\it event horizon} is the union of future and past event horizons.

An initial data set for the Cauchy problem in general relativity consists
of a 3-manifold $\Sigma$, riemannian metric $g_{ab}$, symmetric tensor
$p_{ab}$ (which will be the extrinsic curvature of $\Sigma$ in the evolved globally hyperbolic 
spacetime ${\mathbb R}\times \Sigma$), energy density $\rho$, and momentum
density $J^a$ which satisfies the Hamiltonian and momentum constraints
$$R-p_{ab}p^{ab}+p^2=16\pi\rho \, ,$$
and
$$D_b(p^{ab}-pg^{ab})=8\pi{J^a} \, .$$
Here $R$ is the scalar curvature of the metric $g_{ab}$, $D_b$ is the covariant
derivative defined by $g_{ab}$, and $p\equiv p_a\, ^a$.
Initial data is called physically reasonable if it is smooth (i.e. $C^\infty$), 
$\Sigma$ is geodesically complete with respect to $g_{ab}$, and the sources
obey the local energy condition $\rho \ge {\sqrt {(J_aJ^a)}}$. From here on, initial 
data will always refer to physically reasonable data. When the
energy and momentum densities correspond to classical nondissipative matter
sources or vacuum, $\rho = J^a = 0$, the coupled Einstein-matter equations
can be used to evolve the initial data into a globally hyperbolic spacetime \cite{c,h}. 
Moreover, $\Sigma$ is a spacelike hypersurface 
in the evolved globally hyperbolic 
spacetime, and the constraints are the orthogonal and parallel projections
of the 4-vector arising from contracting the field equations with the normal to 
 $\Sigma$. In the evolved globally hyperbolic spacetime, the local energy 
condition is the dominant 
energy condition, i.e. the stress-energy
tensor satisfies $T_{\alpha \beta}W^\alpha W^\beta \ge 0$, and
$T_{\alpha \beta}W^\beta T^\alpha \, _\gamma W^\gamma \le 0 $ for all
$W^\alpha$ on ${\mathbb R}\times \Sigma $ with $W_\alpha W^\alpha <0$ . 
Usually, $\Sigma$ is required to satisfy the boundary condition that it be a 
closed or asymptotically
flat 3-manifold when describing a cosmological model or an isolated system, respectively. A 3-manifold is {\it closed} if it is compact and has no boundary. 
An {\it asymptotically flat} 3-manifold $\Sigma$ is: a 3-manifold such that for some
compact $C\subset \Sigma$, $\Sigma - C$ consists of a finite number of  
disconnected components each of which is diffeomorphic to ${\mathbb R}^3$
minus a ball $B$. (Note that, the definition of asymptotically flat manifold 
used refers only to differentiable manifolds with no further structure, i.e.
no metric, connection, or other geometric structure.)
Initial data on asymptotically flat 3-manifolds is usually required to satisfy
certain fall off conditions in the asymptotic regions. The most standard
conditions are the following: Initial data on $\Sigma $ is $asymptotically$ 
$flat$ $initial$ $data$, if the metric $g_{ab}$ and the extrinsic curvature 
$p_{ab}$ satisfy ${\hat g}_{ab}- \delta_{ab}= O({1\over r})$, 
$\partial_c{\hat g}_{ab}=O({1\over {r^2}})$, 
$\partial_d \partial_c {\hat g}_{ab}=O({1\over {r^3}})$, 
${\hat p}_{ab}=O({1\over {r^2}})$, 
and $\partial_c {\hat p}_{ab}=O({1\over {r^3}})$ where ${\hat g}_{ab}$ 
 and ${\hat p}_{ab}$ are the pullbacks of $g_{ab}$ and $p_{ab}$ from 
$\Sigma - C$ onto ${\mathbb R}^3 - B$. Another type of initial data on an
asymptotically flat 3-manifold is {\it asymptotically null} initial data. 
This is used when using $p={\rm constant}\not= 0$ on $\Sigma $ and it 
will be discussed later. We will now adopt the convention that initial 
data on asymptotically flat 3-manifolds is asymptotically flat initial data unless 
stated otherwise.  

A {\it CMC hypersurface} in a globally hyperbolic spacetime ${\mathbb R}\times \Sigma $ is a 
spacelike hypersurface diffeomorphic to $\Sigma $ with $p={\rm constant}$. 
A $maximal$ $slice$ is a CMC hypersurface for which $p=0$. Since the 
Hamiltonian constraint must be satisfied on the maximal slice, it follows 
that $R\ge 0$ on $\Sigma $. Therefore, a necessary condition for a globally hyperbolic spacetime 
${\mathbb R}\times \Sigma $ to admit a maximal slice is that $\Sigma $ admit 
a metric with nonnegative scalar curvature. 

Geometrically, a maximal slice is an extremum of the area functional

\begin{equation}\nonumber
A(\Sigma )=\int _{\Sigma } d{\sigma }_g
\end{equation}

\noindent with respect to timelike pushes into the 
globally hyperbolic spacetime ${\mathbb R}\times \Sigma$.
In order to see this take the directional derivative

\begin{equation}\nonumber
\delta _nA(\Sigma )=\delta _n\int _{\Sigma } d{\sigma }_g=\int _{\Sigma } \delta _nd{\sigma }_g
={1\over 2}\int _{\Sigma } g^{ab}\delta _ng_{ab}d{\sigma }_g
\end{equation}
where $n$ is the unit normal to $\Sigma$.

Now, the extrinsic curvature in the evolved spacetime is given by the Lie derivative
\begin{equation}\nonumber
p_{ab}=-{1\over {2N}}{\cal L}_{Nn}g_{ab}=-{1\over {2N}}\delta _ng_{ab}
\end{equation}
\noindent where $N$ is the lapse function. Finally, using this expression for the extrinsic curvature,
one obtains

\begin{equation}\nonumber
\delta _nA(\Sigma )={1\over 2}\int _{\Sigma } g^{ab}\delta _ng_{ab}d{\sigma }_g
=-\int _{\Sigma } Ng^{ab}p_{ab}d{\sigma }_g
\end{equation}

Hence, $\delta _nA(\Sigma )=0$ implies $p=0$. Therefore, the maximal slice is extremal or maximal. CMC slices
can be obtained from a similar variational principle by using a Lagrange multiplier.

The goal is of this section is to construct initial data on a given manifold 
$\Sigma $. This is accomplished 
by showing that every nonnegative function on a closed 3-manifold 
is the energy density of an initial data set with $J^a=0$. In particular, 
every closed 3-manifold has vacuum initial data. Furthermore, there exists 
initial data on every closed 3-manifold so that it is a CMC hypersurface 
for some constant. Next, it will be shown that every asymptotically flat 
$\Sigma $ has initial data. Moreover, every asymptotically flat $\Sigma $ 
has initial data with $p$ equal to a nonzero constant. 

The first case considered will be if $\Sigma $ is a closed 
3-manifold. In order to construct the initial data, we need the following 
theorem due to J. Kazdan and F. Warner \cite{kw}.

\begin{thm}\label{Theorem 1} Given any closed manifold $M^n$ with $n\ge 3$ 
and a smooth function which is negative somewhere on $M^n$, there exists 
a riemannian metric with the prescribed function as its scalar curvature. 
\end{thm}

We will now use this theorem to mimic the initial data for a
Robertson-Walker spacetime. Recall that for Robertson-Walker spacetimes 
the extrinsic curvature $p_{ab}$ is proportional to the metric.

\begin{thm}\label{Theorem 2} Given any closed 3-manifold $\Sigma $ and 
smooth function $\rho $ there exists initial data on $\Sigma $ with energy 
density $\rho $ and $J^a=0$. Furthermore, $p$ is equal to a constant. 
In the case that $\rho =0$, given any nonzero constant $C$ there exists 
vacuum initial data with $p=C$.
\end{thm}

\begin{proof}
Since $\Sigma $ is compact and $\rho $ is smooth, $\rho $ 
attains a maximum on $\Sigma $, call it $\rho _0$. Now, define the smooth 
function $f$ by $f\equiv 16\pi \rho -6A_0^2$ where 
$6A_0^2\equiv 16\pi \rho _0 + \epsilon $ 
and $\epsilon $ is any positive number. The function $f$ is always negative 
on $\Sigma $; and therefore, theorem \ref{Theorem 1} applies to $f$. Let $g_{ab}$ be a 
metric on $\Sigma $ with scalar curvature $f$. Now, define $p_{ab}=A_0g_{ab}$. 
Clearly, $g_{ab}$, $p_{ab}$, $\rho $, and $J^a=0$ form an initial data set on 
$\Sigma $. Taking the trace of $p_{ab}$ yields $p=3A_0={\rm constant}$ which 
completes the first part of the proof. If $\rho =0$, then by rescaling the 
metric by a constant, the scalar curvature can be given any negative value 
$R=-\frac{2}{3} C^2$ where $C$ is an arbitrary constant. Let $g_{ab}$ be the 
metric with $R=-\frac{2}{ 3}C^2$, then $g_{ab}$ and $p_{ab}=\frac{C}{ 3}g_{ab}$ 
satisfy the constraints and $p=C$.
\end{proof}

Having just proven the existence theorem, it is useful to construct an explicit
 concrete example. Begin with the 3-manifold ${\mathbb R}^3$ with the metric 
given by ${\bf g}={dr^2\over {1+k^2r^2}}+r^2d\Omega ^2$ where r is the radial 
coordinate and $d\Omega ^2$ is the standard metric on the unit 2-sphere. 
The scalar curvature of the metric ${\bf g}$ is $R=-6k^2$. This space is 
just hyperbolic 3-space. By identifying points of this space via an 
appropriate group of discrete isometries, it is possible to obtain a closed
manifold. Furthermore, locally the metric of the resulting closed manifold 
is the same as hyperbolic 3-space, because the identifications are done via 
isometries. A particular example of such a manifold is the hyperbolic 
dodecahedron space. It is obtained from identifying opposite faces of a solid 
dodecahedron after a counter-clockwise rotation of $3\pi / 5$ radians. 
One choice of initial data on these spaces is $g_{ab}$, $p_{ab}=Ag_{ab}$, 
$\rho =3 / 8\pi (A^2-k^2)$, and $J^a=0$. 
Moreover, if $\rho $ is 
taken to be the energy density of dust, i.e.   
 $T_{\alpha \beta}=\rho u_\alpha u_\beta$ where $u_\alpha u^\alpha =-1$, 
then the time evolution of the initial data is a Robertson-Walker spacetime 
of negative spatial curvature. If $k^2=A^2$, then the resulting globally hyperbolic spacetime 
is spherically symmetric and vacuum. Thus, it is Minkowski spacetime with points 
identified via a subgroup of the Lorentz group. 

Since the initial data constructed in theorem \ref{Theorem 2} is a generalization of 
Robertson-Walker initial data, a natural choice of a stress-energy tensor 
is that of a dust source. With this in mind, we have the following theorem. 

\begin{thm}\label{Theorem 3} For every closed 3-manifold $\Sigma $, there 
is a spacetime ${\mathbb R}\times \Sigma $ which is physically reasonable, and 
there is at least a single CMC hypersurface. Further, there is a vacuum 
globally hyperbolic spacetime of that form.  
\end{thm}

\begin{proof} If $\rho $ is defined as in the proof of theorem \ref{Theorem 2} and is 
taken to be the energy density of dust, i.e. 
$T_{\alpha \beta}=\rho u_\alpha u_\beta $ where $u_\alpha u^\alpha =-1$, 
then the initial data can be evolved into a globally hyperbolic spacetime. This is proven 
by showing that the coupled Einstein-matter equations form a strictly 
hyperbolic Leray system, and then applying the existence theorems \cite{c,h}. 
If $\rho $ is identically zero, then the evolution exists and gives a globally hyperbolic vacuum 
spacetime for same the reasons as in the dust filled case. Furthermore, in 
both cases the initial data has $p={\rm constant}$, so the initial hypersurface 
is a CMC hypersurface in the evolved globally hyperbolic spacetime. 
\end{proof}

The above theorem took $\rho $  to be the energy density of dust.
 However, one has the same result for  other sources such as the energy-momentum tensor 
for a minimally coupled scalar field $\phi$. It is given by 

$$T_{\alpha \beta } = \nabla_\alpha \phi \nabla_\beta \phi -
1/2 g_{\alpha \beta} \nabla_\gamma \phi \nabla^\gamma \phi\ .$$

Assuming an initially time-independent scalar field, this gives the
energy-momentum current vector 
$J_\alpha=T_{0\alpha}=-1/2 g_{0\alpha} (\nabla \phi)^2$ ($\alpha=1,2,3$). 
One is  free to choose a gauge such that $g_{0\alpha}=0$. With this choice, 
our assumption that $J^\alpha=0$ is satisfied. For the energy density $T_{00}$, we find
  
$$T_{00}=- g_{00} \frac 1 2 (\nabla \phi)^2 = \rho \ .$$

Additionally, a cosmological constant can be also used as a source; pick $\rho $ = constant
and $J^\alpha=0$.

The technique just used to construct initial data on closed 3-manifolds 
cannot be applied to asymptotically flat spacetimes for two reasons: If the 
extrinsic curvature is proportional to the metric, then it will not approach 
zero at infinity and if other asymptotic boundary conditions are applied 
in order to make $p_{ab}=Ag_{ab}$ hold at infinity, then an evolution theorem 
for such initial data sets would have to be proven. Furthermore, one would 
have to show that an asymptotically flat slice also existed in the evolved 
spacetime. Both of these statements are rather difficult to prove and may not 
even hold in general. Thus for asymptotically flat 3-manifolds a new 
procedure is needed for constructing initial data. Since by definition 
every asymptotically flat 3-manifold arises from the removal of a finite 
number of points from a closed 3-manifold, one would like to find out how 
to remove points in such a way that the metric and other fields have the  
correct asymptotic behavior. Motivated by the matching techniques used 
in describing the gravitational collapse of dust, we will remove a finite 
number of balls from a closed 3-manifold with initial data on it, and then 
smoothly glue in a spacelike hypersurface of the Schwarzschild spacetime 
producing the desired asymptotic behavior.

The technique just described will now be used to prove a lemma which 
guarantees the existence of initial data for manifolds admitting a special 
type of metric.

\begin{lem}\label{Lemma 4} Let $S$ be a closed 3-manifold with initial 
data. Suppose that in the neighborhood of some point $i_0$ of $S$ the initial 
data satisfies the following conditions: the metric $g_{ab}$ is spherically 
symmetric with scalar curvature $R=-6k^2$, $p_{ab}=Ag_{ab}$, $J^a=0$, and 
$\rho =3 / 8\pi (A^2-k^2)$.
Then $S-\{ i_0\}$ has asymptotically flat initial data.
\end{lem}

\begin{proof} The spherical symmetry in a neighborhood of $i_0$ means that 
in spherical coordinates centered at $i_0$ the metric takes the form 
${\bf g}=\xi dr^2 + r^2d\Omega ^2$ for $r<r_2$ ($r_2$ fixed) where $\xi $ 
only depends on $r$. Since the metric is geodesically complete and $R=-6k^2$, 
it follows that $\xi ={(1+k^2r^2)}^{-1}$. Now, the goal is to smoothly match 
the above spherically symmetric initial data to the initial data for the 
Schwarzschild spacetime by using a small amount of dust. However, before doing 
this the constraints for general spherically symmetric initial data will be 
written.

A general spherically symmetric metric and extrinsic curvature can be written 
as ${\bf g}=\xi dr^2 + r^2d\Omega ^2$ and ${\bf p}=\alpha dr^2 +\beta r^2d\Omega ^2$ 
where $\xi$, $\alpha$, and $\beta$ only depend on $r$. The constraints can be 
written as 
$$\rho =\frac{1} {16\pi}\Bigl[{{2\xi ^{-2}\xi'}\over r} +{{2(1-\xi^{-1})}\over {r^2}}+4\alpha \beta \xi ^{-1}+2\beta ^2 \Bigr]$$ 
and
$$J^r={1\over {4\pi}}\Bigl[-\beta '\xi ^{-1}+{{\xi ^{-1}}\over r}(\alpha \xi ^{-1}-\beta )\Bigr]\ .$$
Now, assuming that $J^r=0$, the constraints reduce to
$$\rho = \frac{1}  {8\pi r^2}  \frac{d} {dr}\Bigl[\beta ^2r^3-\xi^{-1}r+r\Bigr]$$
and $\alpha =\xi (\beta 'r+\beta )$. Integrating the Hamiltonian constraint 
from $r$ to $r_1$ yields $\beta ^2r^3-\xi ^{-1}r+r+\int_r^{r_1} 8\pi \rho r^2dr\equiv 2C$. Therefore, 
$$\beta ^2={\xi ^{-1}-\Bigl(1-{2M(r)\over r}\Bigr)\over {r^2}}$$ 
where $M(r)\equiv C-\int_r^{r_1} 4\pi \rho r^2dr$.

Now, the match to the initial data for Schwarzschild initial data will be 
performed. Choose any numbers $r_0$ and $r_1$ such that 
$r_2>r_1>r_0>(A^2-k^2)r_1^3$. Let $\eta $ be a smooth monotonically decreasing 
function which is zero for $r_2>r\ge r_1$ and one for $r\le r_0$. Now, choose 
$\rho $ to be equal to ${3\over {8\pi }}(A^2-k^2)(1-\eta )$. At $r=r_1$, the 
quantities $\xi $, $\alpha $, $\beta $, and $\rho $ should correspond to the 
initial data on $S$, and in order for this to hold $C$ must be equal to 
${(A^2-k^2)\over 2}r_1^3$. On the other hand the metric and other expressions 
should become the Schwarzschild initial data for $r\le r_0$ which implies 
$\beta =0$ and $\alpha =0$. This means $\xi ^{-1}=1-{2M(r)\over r}$ for 
$r\le r_0$. Since the metric is required to equal the metric with $R=-6k^2$ for $r\ge r_1$ 
, $\xi ^{-1}$ must equal $1+k^2r^2$ for $r\ge r_1$. In order to smoothly 
interpolate between the two metrics, a natural choice of the metric is 
$$\xi =\Bigl(1-{2M(r)\over r}\Bigr)^{-1}\eta +(1-\eta)\bigl(1+k^2r^2\bigr)^{-1}$$
for $r_0\le r\le r_1$. Both $\xi $ and the induced $\beta ^2$ are smooth and 
positive. Further, they go smoothly to the appropriate functions at the points 
$r_0$ and $r_1$. Finally, choosing $\alpha = \xi (\beta 'r+\beta)$ the 
constraints are satisfied by construction, and the initial data is smooth. 
Moreover, the initial data is equal to the Schwarzschild initial data for 
$r\le r_0$ and the initial data on $S$ for $r\ge r_1$. Therefore, in order 
to obtain asymptotically flat initial data, smoothly extend the Schwarzschild 
initial data across the throat at $r=2M(r_0)$.
\end{proof}

The lemma just proven only applies to closed 3-manifolds possessing a special 
type of metric. Examples of 3-manifolds satisfying the hypothesis of the 
lemma are closed flat and hyperbolic 3-manifolds. It is clear that if a closed 
3-manifold has a finite number of points at which the conditions are satisfied,
 then the gluing procedure can be applied to each of the points to produce an 
asymptotically flat 3-manifold with many asymptotic regions. 
Since the definition of asymptotic flatness implies that every asymptotically 
flat manifold arises from the removal of a finite number of points from a 
closed manifold, we only need to prove every closed 3-manifold has initial 
data satisfying the hypothesis of lemma \ref{Lemma 4} in order to prove that 
every asymptotically flat 3-manifold has initial data. These ideas are employed
 to prove the following general existence theorem.

\begin{thm}\label{Theorem 5} Every asymptotically flat 3-manifold has initial data.
\end{thm}

\begin{proof} By definition, every asymptotically flat 3-manifold $\Sigma $ 
has a compact subset $C$ such that $\Sigma -C$ has a finite number disconnected 
components each of which is diffeomorphic to ${\mathbb R}^3$ minus a ball $B$. One 
can easily compactify ${\mathbb R}^3 - B$ in a smooth way to obtain $S^3-B$. Using 
this compactification for each asymptotic region of $\Sigma $, it follows that 
$\Sigma $ is diffeomorphic to a closed 3-manifold $\widetilde \Sigma$ minus a 
finite number of points. (Remember that when we refer to manifolds, they only 
have a differentiable structure and no other structure.) This means that every 
asymptotically flat 3-manifold $\Sigma $ arises by removing points from a 
closed 3-manifold $\widetilde \Sigma $.

Now, let $\widetilde \Sigma $ be a closed 3-manifold and $\{x_1,x_2,\dots ,x_n\}$ 
be any set of a finite number of points in $\widetilde \Sigma $. Pick any metric $g_{ab}$ 
on $\widetilde \Sigma $. Next, pick a neighborhood $N_j$ about each point $x_j$ 
such that $N_j$ is diffeomorphic to a ball in ${\mathbb R}^3$ with radius $r_3$ 
via a diffeomorphism ${\psi _j}\colon N_j\to B$. Choose a smooth monotonically 
decreasing function $\eta $ which is equal to zero for $r_3>r>r_2$ and one 
for $r_2>r_1>r$. Let ${\tilde g}_{ab}^j$ be the metric defined on $B$ by 
$${\tilde g}_{ab}^j\equiv (1-\eta){\psi _j^{-1}}g_{ab}|_{N_j} + \eta h_{ab}$$
where $\psi _j^{-1}g_{ab}|_{N_j}$ is the pullback of $g_{ab}|_{N_j}$ onto 
$B$ and $h_{ab}$ is the spherically symmetric metric with constant scalar 
curvature $-6k^2$ on ${\mathbb R}^3$. Finally, define a new metric ${\bar g}_{ab}$ 
on $\widetilde \Sigma $ to be equal to $g_{ab}$ on ${\widetilde \Sigma }- \bigcup _jN_j$ and $\psi _j {\tilde g}^j_{ab}$ on each $N_j$. Then ${\bar g}_{ab}$ is 
spherically symmetric and has constant scalar curvature in neighborhoods of 
each of the points $x_j$. 

Next, initial data will be constructed. Let $\overline R$ be the scalar curvature
 of ${\bar g}_{ab}$. The scalar curvature ${\overline R}$ is smooth and 
$\widetilde \Sigma $ is compact, this means $\overline R$ attains a minimum 
denote it by ${\overline R}_0$. Now, define a smooth function 
$\rho \equiv {{{\overline R}+6A^2}\over {16\pi }}$ where 
$6A^2\equiv |{\overline R}_0| + \epsilon$, $\epsilon >0$. Since $\rho \ge 0$, 
the choice $J^a=0$ means the local energy condition is satisfied. Choose 
${\bar g}_{ab}$, $p_{ab}=A{\bar g}_{ab}$, $J^a=0$, and 
$\rho = {{{\overline R}+6A^2}\over {16\pi }}$ as initial data on $\widetilde \Sigma $.

Finally, apply lemma \ref{Lemma 4} at each point in $\{ x_1,x_2,\dots ,x_n\}$. This can be 
done because the proof of lemma \ref{Lemma 4} only uses the local structure of the metric 
and other fields. Therefore, ${\widetilde \Sigma }-\{ x_1,x_2,\dots ,x_n\}$ 
has asymptotically flat initial data.
\end{proof}

The initial data constructed in the above theorem has $p={\rm constant}$ 
in a compact region but vanishing in the asymptotic regions. However, it is  
sometimes useful to use initial data with $p$ constant everywhere, for 
example when studying the behavior of singularities \cite{e}. Although, 
it will be shown that initial data with $p$ everywhere zero does not always 
exist for all manifolds, the following theorem is an existence theorem for 
initial data with $p$ equal to a nonzero constant on an arbitrary asymptotically
 flat manifold. For initial data of this form, different asymptotic boundary 
conditions are imposed on the metric and other fields in 
order that the Hamiltonian constraint hold at infinity. Initial data with 
$p$ equal to a constant is a particular example of {\it asymptotically} 
{\it null} initial data. The reason for the terminology is that the hypersurface
 reaches null infinity in the evolved spacetime. Usually, one requires the 
metric and extrinsic curvature to approach the metric and extrinsic curvature 
of a CMC hypersurface of the Schwarzschild spacetime. However, at the present 
time there are no standardized conditions on the rate at which the metric and 
extrinsic curvature approach the Schwarzschild asymptotically null initial data,
 so we will prove existence of asymptotically null initial data with arbitrary 
topology and $p={\rm constant}$ under the strongest conditions, namely, that 
outside a compact set the initial data in each of the asymptotic regions is 
equal to the initial data of a CMC hypersurface in the Schwarzschild spacetime. 
  
\begin{thm}\label{Theorem 6} Every asymptotically flat 3-manifold has 
asymptotically null initial data with $p={\rm constant}$.
\end{thm}

\begin{proof} Let $\widetilde \Sigma $ be a closed 3-manifold and $\{ x_1,x_2,\dots ,x_n\}$ 
a finite set of points in $\widetilde \Sigma$. Now, let ${\bar g}_{ab}$ be the 
metric defined in the proof of theorem \ref{Theorem 5}. Pick the initial data on 
$\widetilde \Sigma $ to be the initial data from the proof of theorem \ref{Theorem 5}, i.e. 
${\bar g}_{ab}$, $p_{ab}=A{\bar g}_{ab}$, $J^a=0$, and 
$\rho ={{{\overline R}+ 6A^2}\over {16\pi }}$.
Recall the initial data is spherically symmetric in neighborhoods of each of 
the points $\{ x_1,x_2,\dots ,x_n\}$. The goal is to match this initial data 
to the Schwarzschild initial data for a CMC hypersurface. Before doing this, 
the constraints for spherically symmetric initial data with $p$ equal to a 
constant will be derived.

Given a spherically symmetric metric and extrinsic curvature, they can be 
written as ${\bf g}=\xi dr^2+r^2d\Omega ^2$ and ${\bf p}=\alpha dr^2 + \beta r^2d\Omega ^2$, 
where the coefficients only depend on $r$. The constraints can be written as 
$$\rho = {1\over {16\pi }}{\Bigl [ {{2\xi ^{-2}\xi '}\over r}+{{2(1-\xi ^{-1})}\over {r^2}}+ 4\alpha \beta \xi ^{-1} + 2\beta ^2\Bigr ]}$$
and 
$$J^r = {1\over {4\pi }}{\Bigl [-\beta '\xi ^{-1} + {{\xi ^{-1}}\over {r}}{(\alpha \xi ^{-1} - \beta )}\Bigr ]}\ .$$
Now, assuming that the extrinsic curvature is proportional to the metric, i.e. 
$p_{ab}=Ag_{ab}$, it follows that $J^r=0$ and the Hamiltonian constraint 
becomes 
$$\rho = {1\over {16\pi }}\Bigl[{{{2\xi ^{-2}}\xi '}\over r}+{{2(1-\xi ^{-1})}\over {r^2}} +6A^2\Bigr]={1\over {8\pi r^2}}\Bigl[1-{d\over {dr}}\bigl(r\xi ^{-1}\bigr)\Bigr]+{3\over {8\pi}}A^2\ .$$    
In our case, pick one of the points $x_j$ for which the the initial data on 
$\widetilde \Sigma $ is spherically symmetric, then $\rho ={3\over {8\pi }}\bigl(A^2-k^2\bigr)$ 
in a neighborhood of $x_j$. As in lemma \ref{Lemma 4}, let $r_2$ be a fixed number for 
which the metric and extrinsic curvature can be expressed as 
${\bf g}=\xi dr^2 +r^2d\Omega ^2$ and ${\bf p}=\alpha dr^2 +\beta r^2d\Omega ^2$
 for all $r<r_2$. Choose any numbers $r_0$ and $r_1$ such that 
$r_2>r_1>r_0>\bigl (A^2-k^2\bigr )r_1^3$. For the above initial data 
$\rho ={3\over {8\pi }}\bigl(A^2-k^2\bigr)$ for $r<r_2$. Let $\gamma $ be a 
smooth monotonically increasing function which is equal to one for $r>r_1$ 
and zero for $r<r_0$. Now, define ${\bar \rho}\equiv \rho \gamma $. 
Substituting $\bar \rho $ into the Hamiltonian constraint and integrating 
from $r$ to $r_1$ yields $\xi ^{-1}= 1-{2M(r)\over r} +A^2r^2$ where 
$M(r)\equiv {{\bigl(A^2-k^2\bigr)}\over 2}r_1^3- \int_r^{r_1}4\pi {\bar \rho}r^2dr$. 
By definition $\xi $ is smooth, and it is easily shown to be positive. For 
$r>r_1$, $\xi ^{-1}=1+k^2r^2$. Now, if $r<r_0$, then $\xi ^{-1}=1-{2M({r_0})\over r}+A^2r^2$ 
which is the metric for a CMC hypersurface in the Schwarzschild spacetime. 
Finally, extend $\xi $ across the throat by using Kruskal coordinates. 
This gives an asymptotic region as in the proof of lemma \ref{Lemma 4} . Since the above 
arguments were all only local, they can be applied at each of the points 
$\{ x_1,x_2,\dots ,x_n\}$. Therefore, there is initial data on 
${\widetilde \Sigma}-\{ x_1,x_2,\dots ,x_n\}$ with $p=3A$.
\end{proof}

Since the initial data constructed in theorems \ref{Theorem 5} and \ref{Theorem 6} has momentum density 
equal to zero, $J^a=0$, one choice for the matter source is dust, and just as 
in the case of the closed manifolds the initial data can be evolved into a globally hyperbolic 
spacetime. Each asymptotic region of one these evolved spacetimes is just the 
asymptotic region of the Schwarzschild spacetime, while the interior region 
is a piece of one of the spatially closed globally hyperbolic spacetimes constructed 
in theorem \ref{Theorem 3}. 
Using dust as the source of $\rho $ yields the following theorem.

\begin{thm}\label{Theorem 7} For every asymptotically flat 3-manifold 
$\Sigma $, there is a globally hyperbolic spacetime ${\mathbb R}\times \Sigma $ which is physically 
reasonable. Furthermore, one can always find a globally hyperbolic spacetime ${\mathbb R}\times \Sigma $
 with a CMC hypersurface. 
\end{thm}

\begin{proof}
Choose as the matter source $T_{\alpha \beta}=\rho u_\alpha u_\beta $ where $u_\alpha u^\alpha =-1$. 
Then invoking the same existence theorems as in theorem \ref{Theorem 3} and using the initial 
data of theorem \ref{Theorem 5} yields the first part 
of the theorem. To obtain the second result just repeat the procedure using 
the initial data from theorem \ref{Theorem 6}. 
\end{proof} 

The fact that all closed and asymptotically flat 3-manifolds are allowed by 
the classical field equations as the spatial topologies of globally hyperbolic spacetimes will 
be used in the next section to prove the nonexistence of maximal slices in 
general. The nonexistence of maximal slices for these globally hyperbolic spacetimes only depends 
on the topology of the 3-manifold and not on the matter sources. The results 
of this section also show that CMC hypersurfaces for some $p\not= 0$ exist 
for all topologies. Hence, there are no topological obstructions to general 
CMC hypersurfaces.  
 
\section{ TOPOLOGICAL OBSTRUCTIONS}

Suppose that a globally hyperbolic spacetime ${\mathbb R}\times \Sigma $ has a maximal slice and 
obeys the dominant energy condition, then $p=0$ and $\rho \ge 0$ on the slice. 
Furthermore, the constraints must also hold on the slice, in particular 
$R={16\pi }+ p_{ab}p^{ab} - p^2$. Combining these facts together implies that 
$\Sigma $ has a metric with $R\ge 0$. Therefore, a necessary condition for the 
existence of a maximal slice is that $\Sigma $ admit a metric with $R\ge 0$. 
Using a natural method of counting 3-manifolds, the set of all closed 
3-manifolds which admit a metric with $R\ge 0$ comprise 
a small fraction of all closed 3-manifolds. Since theorem \ref{Theorem 3} proves that all 
closed 3-manifolds occur as hypersurfaces of globally hyperbolic spacetimes, it follows that 
only a small fraction of spatially closed globally hyperbolic spacetimes admit maximal slices. 
For asymptotically flat globally hyperbolic spacetimes, a maximal slice implies that $R\ge 0$ 
on the asymptotically flat 3-manifold. It can be shown that an asymptotically 
flat 3-manifold with $R\ge 0$ has a smooth compactification with $R>0$; and 
therefore, only a small fraction of asymptotically flat globally hyperbolic spacetimes 
have maximal slices. 
The results for closed 3-manifolds will now be discussed, and the compactification 
theorem will be proven at the end of this section. 

Given a closed manifold there can be topological obstructions to placing a 
metric on it with $R\ge 0$. The most well known example of such an obstruction 
is the Euler characteristic in the case of 2-manifolds. The idea is assume 
some 2-manifold has a metric with $R\ge0$, then integrate the scalar curvature 
over $M$. Next, applying the Gauss-Bonnet theorem $\chi (M)={1\over {4\pi}}\int_M RdA$, 
it follows that $\chi (M)\ge 0$. There are only four closed 2-manifolds 
 with nonnegative Euler characteristic, namely, the sphere, projective plane, 
torus, and Klein bottle. Since there are a countable number of distinct closed
 2-manifolds and only four of them which admit metrics with $R\ge0$, this 
proves not only are there obstructions to admitting a metric with $R\ge0$ 
but also that manifolds with $R\ge0$ are rare. 

An example of a topological obstruction which is more closely related to the 
obstruction for admitting a metric with $R\ge 0$ on manifolds of dimension 
greater than two is the first Betti number. More precisely, if a closed 
manifold has a metric with positive Ricci curvature, then the first Betti 
number vanishes. First, define $\Omega ^p({M^n})$ to be the vector space 
of all p-forms on $M^n$. We have the following elliptic complex 

\begin{align}
& 0\ar \Omega ^0(M^n)\dar \Omega ^1(M^n) \dar \dots\dar
\Omega ^p(M^n) \nonumber \\  & \dar \Omega^{p+1}(M^n) \dar\dots\dar\Omega
^{n-1}(M^n)\dar\Omega ^n(M^n)\ar 0
\end{align}

where $d$ is the exterior derivative. The $p^{th}$ cohomology $H^p({M^n})$ 
is defined by 
$${Ker\bigl(\Omega ^p(M^n)\dar\Omega ^{p+1}(M^n)\bigr)}/{Im\bigl(
\Omega ^{p-1}(M^n)\dar\Omega ^p(M^n)\bigr)}\ ,$$ 
in other words, it is the space of closed p-forms modulo exact p-forms. For closed 
n-manifolds de Rham's theorem says that $H^p(M^n)$ is isomorphic to the real 
singular cohomology, i.e. the usual cohomology which is calculated using only 
the topology of $M^n$. Picking a metric on $M^n$ allows the choice 
of a unique representative of each cohomology class. In order to choose the 
unique representative, first define an inner product of forms 
$$\bigl(\alpha ,\beta \bigr)=\int_{M^n}\, \alpha_{ab\dots c}\beta ^{ab\dots c}\,dV\ .$$
The adjoint of $d$ with respect to this inner product denoted by $\delta $ can 
be written as $\delta \alpha =-\nabla ^a\alpha _{ab\dots c}$. Now, define 
the energy of a form to be 
$$E(\alpha )=\int_{M^n}\, \alpha _{ab\dots c}\alpha ^{ab\dots c}\, dV\ .$$
Given a cohomology class $\bigl[\alpha _0\bigr] \in H^p({M^n})$, then the class 
 can be represented by $\alpha = \alpha _0 + d\beta $ where $\beta $ ranges 
over all (p-1)-forms. The goal is given a fixed $\alpha _0$ find a $\beta $ 
such that $\alpha $ has least energy. Taking the variation, the Euler-Lagrange 
equations imply that $\delta \alpha =0$. Since $\alpha $ was closed to begin 
with, it follows that $\delta \alpha =0$ and $d\alpha =0$. Observe that for 
2-forms the above equations are just the vacuum Maxwell equations for a positive
 definite metric. Let the Laplacian of a p-form $\alpha $ be defined by  
$\Delta \alpha= (\delta d + d\delta )\alpha$. The conditions that $\delta \alpha=0$ 
and $d\alpha =0$ are equivalent to $\Delta \alpha =0$, because if $\alpha $ is $harmonic$ then  
$$0=\bigl(\Delta \alpha ,\alpha \bigr)=\bigl(\delta\alpha ,\delta\alpha\bigr)+ \bigl(d\alpha ,d\alpha \bigr) \ .$$
The above arguments imply that the kernel of $\Delta _p$ or the space of 
harmonic 
p-forms is isomorphic to the $p^{th}$-cohomology vector space $H^p(M^n)$; this 
result is known as Hodge's theorem. 

Finally, Hodge theory is applied to the problem of finding obstructions to 
admitting metrics with positive Ricci curvature, $R_{ab}\xi ^a\xi ^b>0$ 
if $\xi $ is not zero everywhere. One can verify that 
$$\Delta \xi _a=-\nabla ^2\xi _a + R_a\, ^b \xi _b$$ 
by using the facts that $d\omega =\nabla _{[a}\omega _{bc\dots d]}$ and 
$\delta \alpha =-\nabla ^a\alpha _{ab\dots c}$. Suppose that 
$\xi \in ker\Delta $, then 
$$0=\int_{M^n}\, \xi ^a\Delta \xi _a\, dV= -\int_{M^n}\, 
\xi ^a\nabla ^2\xi _a\, dV + \int_{M^n}\, R_{ab}\xi ^a\xi ^b\, dV$$ 
but 
$$-\int_{M^n}\, \xi ^a\nabla ^2\xi _a\, dV+\int_{M^n}\, R_{ab}\xi ^a\xi ^b\, dV
=\int_{M^n}\, \nabla ^b\xi ^a\nabla _b\xi _a\, dV+\int_{M^n}\, R_{ab}\xi ^a\xi ^b\, dV>0$$ 
assuming positive Ricci curvature. This means $\xi =0$. Therefore, positive Ricci 
curvature implies that the kernel of $\Delta $ is trivial and that $H^1(M^n)$ 
is zero. Since the first Betti number $b_1(M^n)$ is the dimension of $H^1(M^n)$,
 this result is equivalent to $b_1(M^n)=0$. For the n-torus $b_1(T^n)=n$, 
therefore the n-torus admits no metric with positive Ricci curvature. Another 
example of a manifold which admits no metric with positive Ricci curvature 
is any closed manifold of the form $N\times T^n$ where $N$ is an arbitrary 
manifold. This is true because $b_1(N\times T^n)=b_1(N)+b_1(T^n)=b_1(N)+n>0$. 
With the above examples in mind, the case of positive scalar curvature will 
now be considered. 

The key to finding an obstruction to positive Ricci curvature was finding 
a generalized Laplacian which split into the {\it usual} Laplacian and a linear 
term involving the curvature. The object is to find an operator which satisfies 
 these conditions with the linear term being the scalar curvature. 
Given a manifold 
with spinors on it, i.e. a spin manifold, there is such an operator, namely, 
the Dirac operator, denoted by $\drc$. The {\it Weitzenb\"ock formula} for the 
Dirac operator is $\drc^{\, 2}\psi =-\nabla ^2\psi +{1\over 4}R\psi $. 
The operator $\drc^{\, 2}$ is the 
Laplacian operator for spinors, and the kernel of $\drc^{\, 2}$ is the space of 
{\it harmonic spinors}. If $R>0$, then the Weitzenb\"ock formula implies that 
$ker\, \drc^{\, 2}=0$. However, a difficulty with the 
Dirac operator is that there 
is no analogue of Hodge's theorem. In fact, the space of harmonic spinors in 
general depends on the metric. To avoid this problem, one might consider the 
index of the Dirac operator or an operator constructed from it because the 
Atiyah-Singer index theorem guarantees that the index is a topological 
invariant. However, in the case of interest, namely 3-manifolds, all such 
invariants vanish, so no information on which 3-manifolds admit metric with 
$R>0$ can be obtained directly. In order to over come this difficulty, replace 
the closed 3-manifold $\Sigma $ with the 4-manifold $\Sigma \times S^1$. Observe
 that although such a manifold admits no metric with positive Ricci curvature, 
the scalar curvature is much weaker; and in fact, if $\Sigma $ admits a metric 
with $R>0$, then $\Sigma \times S^1$ admits a metric with $R>0$. To prove this, 
let the metric on $\Sigma \times S^1$ be given by $ds^2=g_{ij}dx^idx^j+d\theta ^2$ 
where $g$ is the metric on $\Sigma $ which has $R>0$, then this product metric 
has positive scalar curvature. Now, look at the Dirac operator on closed 
4-manifolds with $R>0$. First, on any even dimensional manifold the bundle of 
spinors $S$ can be decomposed into the sum of $S_+\bigoplus S_-$, namely, the 
spinors of $+{1\over 2}$-helicity and $-{1\over 2}$-helicity. Now, define the 
operator $\drc^{\, +}$ to be the restriction of $\drc$ to $S_+$. 
It is easily proven that $\drc^ {\, +}$ maps $S_+$ into $S_-$ and that the 
adjoint of $\drc ^{\, +}$ is the restriction of $\drc$ to $S_-$ 
which is denoted by $\drc^{\, -}$. The index of a differential operator is 
the dimension of the kernel minus the dimension of the cokernel of the 
operator. When the operator has an adjoint the cokernel of the operator is the 
kernel of the operator's adjoint. Thus, the index of 
$\drc^{\, +}$ is $ind(\drc^{\, +})=dim\,ker\, \drc^{\, +} - dim\, ker\, 
\drc^{\, -} $.  
Since the 4-manifold being considered has $R>0$, then $ker\, \drc^{\, 2}=0$ but 
$\drc$ is self-adjoint so $ker\, \drc=ker\, \drc^{\, 2}=0$. 
Furthermore, $S=S_+\bigoplus S_-$ implies
 that $ker\, \drc=ker\, \drc^{\, +}\bigoplus ker\, \drc^{\, -}$ which 
means $R>0$ implies $ker\, \drc^{\, +}=0$ 
and $ker\, \drc^{\, -}=0$. Hence, $R>0$ implies $ind(\drc^{\, +})=0$. 
Identifying the topological invariant associated with $ind(\drc^{\, +})=0$ 
is a rather involved procedure, so the 
details will not be given here. The invariant is the $\widehat A$-genus and in 
four dimensions it is proportional to the signature. The {\it signature} of a 
closed 4n-manifold $M^{4n}$ is the signature of the bilinear form 
$(\alpha ,{}^*\beta )$ where $\alpha ,\beta \in H^{2n}\bigl(M^{4n}\bigr)$, 
${}^*\beta $ is the dual form of $\beta$, and $(\ ,\ )$ is the inner product of 
forms defined above. The signature of closed manifolds of other 
dimensions is taken to be zero. Now, with the 
invariant in hand let us return to the manifold $\Sigma \times S^1$; the 
signature $\tau $ of $\Sigma \times S^1$ is 
$\tau (\Sigma \times S^1)=\tau (\Sigma )\tau (S^1)$. Since 
$\tau (\Sigma)=\tau (S^1)=0$, it follows that $\tau (\Sigma \times S^1)=0$ 
which implies the 
$\widehat A$-genus vanishes. The problem of the $\widehat A$-genus vanishing 
is caused by the fact that the signature obeys the product rule. In order 
to obtain a nontrivial invariant, M. Gromov and H. B. Lawson define a family of 
generalized Dirac operators on general bundles of spinors over $\Sigma \times 
S^1$. These generalized Dirac operators still satisfy generalized Weitzenb\"ock 
formulae which imply; if $R>0$, then the generalized $\widehat A$-genera all 
must vanish. The difference between the classical $\widehat A$-genus and the 
generalized $\widehat A$-genera is that the generalized ones do not satisfy 
the product rule. Using the generalized $\widehat A$-genera M. Gromov and H. B. 
Lawson were able to prove the following classification theorem \cite{g}.

\begin{thm}\label{Theorem 8} A closed orientable 3-manifold $M^3$ (if it is 
nonorientable, then take its double cover) which has a $K(\pi ,1)$ as 
a prime factor\footnote{A $K(\pi ,1)$ is a closed 3-manifold with a contractible
universal covering space. For example, the 3-torus is a $K(\pi
,1)$. See reference \cite{jh} for more details.}
 in its prime decomposition\footnote{The {\it connected sum} of two 3-manifolds $M_1$ and
$M_2$ is the 3-manifold $M_1\#M_2$ obtained from removing a ball from
each, and then gluing the resulting manifolds together along their
boundaries. A closed 3-manifold is {\it prime} if $M=M_1\#M_2$ implies
that $M_1$ or $M_2$ is a 3-sphere. Every 3-manifold has a unique prime
decomposition into the connected sum of a finite number of prime
factors. Also see reference \cite{jh}.}
admits no metric 
with $R>0$. In fact, any metric on $M^3$ with $R\ge 0$ must be flat. 
\end{thm}

The work of W. Thurston \cite{t} on hyperbolic 3-manifolds implies most 
3-manifolds are in fact $K(\pi ,1)$'s in the following way: First, 
a {\it knot} is the continuous embedding of a circle in a 3-manifold and a 
{\it link} is a finite number of disjoint knots. Next, one 
preforms {\it Dehn surgery} 
 along a link in a manifold by removing tubular neighborhoods of each knot, 
and then gluing back the removed solid tori differently. More precisely, one 
identifies the boundary of each hole left  with the boundary of
 another solid torus via a homeomorphism of the boundary different from 
the one defined by the inclusion of the removed torus in the manifold.  
W. B. R. Lickorish proved that every closed orientable 3-manifold can be 
obtained from Dehn surgery on the 3-sphere. More recently, W. Thurston has 
proven that every closed orientable 3-manifold is obtained from the 3-sphere 
$S^3$ by Dehn 
surgery along links $L$ such that the 3-sphere minus $L$ is a hyperbolic 
3-manifold, call such links {\it hyperbolic}. Furthermore, all but a finite 
number of 3-manifolds obtained obtained from $S^3$ by Dehn surgeries along a 
given hyperbolic link $L$ are closed hyperbolic 3-manifolds. Hence, in this way
 of counting 3-manifolds most closed 3-manifolds are hyperbolic. 
Since all hyperbolic manifolds are covered by ${\mathbb R}^3$, it follows that they 
are all $K(\pi ,1)$'s. Therefore, most 3-manifolds are $K(\pi ,1)$'s.  
 
The above result combined with theorem \ref{Theorem 8} implies most 3-manifolds never admit a 
metric with $R>0$. Furthermore, there are only 
ten flat closed 3-manifolds which means most closed 3-manifolds do not admit 
metric with $R\ge 0$. Combining these ideas with the existence theorems for 
globally hyperbolic spacetimes with arbitrary spatial topology yields the following result. 

\begin{thm}\label{Theorem 9} Most spatially closed globally hyperbolic 
spacetimes ${\mathbb R}\times \Sigma $ never admit a maximal slice. 
\end{thm}

\begin{proof} Existence of a maximal slice implies $\Sigma $ admits a metric 
 with $R\ge 0$, and theorem \ref{Theorem 3} implies all closed 3-manifolds occur as 
hypersurfaces. Therefore, theorem \ref{Theorem 8} implies most 
globally hyperbolic spacetimes ${\mathbb R}\times \Sigma $ 
admit no maximal slice. 
\end{proof}

Finally, the classification of asymptotically flat 3-manifolds which admit 
metrics with $R\ge 0$ will be discussed. This classification theorem is proven 
by showing that every asymptotically flat 3-manifold with a metric having 
$R\ge 0$ has a smooth compactification with $R>0$. The classification is 
completed by applying theorem \ref{Theorem 8} to the compactifications. Before proving the 
compactification theorem, several technical propositions are required which 
are now presented. 

The first technical proposition needed is the maximum principle for second order
 elliptic operators. One important feature of this version of the maximum principle is there 
are no restrictions on the sign of $c$, the coefficient of the zeroth order term
 of the operator. The maximum principle will not be proven but the interested 
reader is referred to M. Spivak's book on Differential Geometry \cite{s}. 

\begin{thm}\label{Theorem 10} Let $L$ be a elliptic differential operator defined by 
$$Lu= \sum_{i,j=1}^n \, a_{ij}{{\partial^2u}\over {\partial x_i\partial x_j}} + 
\sum_{i=1}^n\, b_i{{\partial u}\over {\partial x_ i}} + cu$$
on an open subset $U \subseteq {\mathbb R}^n$. The coefficients $a$, $b$, and $c$ 
are locally bounded; a is symmetric; and in the neighborhood of any point of 
$U$ there are two positive constants $m$ and $M$ such that 
$$m\sum_{i=1}^n\xi _i^2\le \sum_{i,j=1}^n\, a_{ij}\xi _i\xi _j\le M\sum_{i=1}^n\xi _i^2$$ 
for all $\xi \in {\mathbb R}^n$. If $u\in C^2(U)$ with $Lu\ge 0$ and $u\le 0$, then 
 $u(x_0)=0$ for some $x_0\in U$ implies $u\equiv 0$ in $U$. 
\end{thm}

The next proposition is an existence lemma for a family of ``bump functions'' 
 with a prescribed type of asymptotic behavior. The existence of such a one 
parameter family of functions is an essential ingredient in the proof of the 
compactification.

\begin{lem}\label{Lemma 11} Given $c>0$, there is a family of smooth decreasing 
functions $\alpha_\rho$ for $\rho >c$ such that $\alpha _\rho =1$ for $r\le \rho $,  
$\alpha _\rho =0$ for $r\ge 2\rho $, $|\alpha _\rho '|\le {A\over r}$, and 
$|\alpha _\rho ''|\le {A\over {r^2}}$ where $A$ is a constant independent of 
$\rho $ and $\alpha _\rho '\equiv {d\over {dr}}\alpha _\rho (r)$. 
\end{lem}

\begin{proof}
 Let $\gamma _{{}_L}$ be any smooth function which is equal to 
one for $\rho +\epsilon \le r\le {4\over 3}\rho -\epsilon$ where $\epsilon =
{1\over 6}c$, and zero otherwise. 
Likewise, $\gamma _{{}_R}$ is one for ${5\over 3}\rho +\epsilon \le r\le 2\rho -\epsilon $ 
and zero otherwise. Now, define smooth increasing and decreasing functions by 
the following expressions 
$$\alpha _{{}_L}(r)\equiv {{\int_\rho ^r{{\gamma _{{}_L}}\over t}dt}\over {\int _\rho ^{{4\over 3}\rho }{{\gamma _{{}_L}}\over t}dt}}$$ 
and 
$$\alpha _{{}_R}(r)\equiv {{\int_r^{2\rho }{{\gamma _{{}_R}}\over t}dt}\over {\int_{{5\over 3}\rho }^{2\rho }{{\gamma _{{}_R}}\over t}dt}}\ .$$
Observe that $|\alpha _{{}_L}'|\le \bigl({\int _\rho ^{{4\over3}\rho }{{\gamma _{{}_L}}\over t}dt}\bigr)^{-1}r^{-1}$ 
and $|\alpha _{{}_R}'|\le \bigl({\int _{{5\over 3}\rho}^{2\rho}{{\gamma _{{}_R}}\over t}dt}\bigr)^{-1}r^{-1}$. 
Furthermore, 
$$\int _\rho ^{{4\over 3}\rho}{{\gamma _{{}_L}}\over t}dt\ge \int _{\rho +\epsilon }^{{4\over 3}\rho -\epsilon}{{dt}\over t}=\log {{{4\over 3}\rho -\epsilon}\over {\rho +\epsilon }}\ge \log {4\over 3}$$  
and 
$$\int_{{5\over 3}\rho }^{2\rho }{{\gamma _{{}_R}}\over t}dt\ge\int_{{5\over 3}\rho +\epsilon }^{2\rho -\epsilon }{{dt}\over t}=\log {{2\rho -\epsilon }\over {{5\over 3}\rho +\epsilon }}\ge \log {6\over 5}\ .$$ 
Hence the derivatives of $\alpha _{{}_L}$ and $\alpha _{{}_R}$ having the following behavior 
$|\alpha _{{}_L}'|\le \bigl({\log {4\over 3}}\bigr)^{-1}r^{-1}$ and 
$|\alpha _{{}_R}'|\le \bigl({\log {6\over 5}}\bigr)^{-1}r^{-1}$. Next, define 
another smooth function by the following expression

\begin{displaymath}
\gamma _\rho (r)\equiv 
\begin{cases}
\alpha _{{}_L}(r), &if \rho \le r\le \frac{4} {3}\rho ;\\ 1, &if
\frac4 3 \rho \le r \le {5\over 3}\rho ;\\ \alpha _{{}_R}(r), &if
\frac 5 3 \rho \le r\le 2\rho;\\ 0, &otherwise.
\end{cases}
\end{displaymath}

It follows that $|\gamma _\rho '|\le \bigl({\log {4\over 3}}\bigr)^{-1}r^{-1}$. 
Finally, let 
$$\alpha _\rho (r)\equiv {{\int_r^{2\rho }{{\gamma _\rho }\over t}dt}\over {\int_\rho ^{2\rho}{{\gamma _\rho}\over t}dt}}\ .$$ 
By definition $\alpha _\rho $ is smooth, $\alpha _\rho =1$ for $r\le \rho $, 
and $\alpha _\rho =0$ for $r\ge 2\rho $. Observe that 
$$\int_\rho ^{2\rho }{{\gamma _\rho }\over t}dt\ge \int_{{4\over 3}\rho }^{{5\over 3}\rho }{{dt}\over t}=\log {5\over 4}\ ,$$
this implies $|\alpha _\rho '|\le \bigl({\log {5\over 4}}\bigr)^{-1}r^{-1}$. Furthermore, 
$$|\alpha _\rho ''|\le \bigl({\log {5\over 4}}\bigr)^{-1}\Bigl[{{|\gamma '|}\over r}+{{|\gamma |}\over {r^2}}\Bigr]\le \bigl({\log {5\over 4}}\bigr)^{-1}\Bigl[\bigl({\log {4\over 3}}\bigr)^{-1}+1\Bigr]r^{-2}\ .$$ 
Therefore, $|\alpha _\rho '|\le {A\over r}$ and $|\alpha _\rho ''|\le {A\over r^2}$ 
where $A= \bigl({\log {5\over 4}}\bigr)^{-1}\bigl[\bigl({\log {4\over 3}}\bigr)^{-1} +1\bigr]$. 
\end{proof}

The next lemma is an existence theorem for the ground state of the operator 
$-D^2 + V$. This lemma is used in constructing a metric with positive scalar 
curvature on the compactified manifold.   

\begin{lem}\label{Lemma 12} If $M$ is a closed n-manifold and $L= -D^2 + V$ where 
$V$ is a smooth function on $M$, then there is a smooth function $\psi _0>0$ 
and real number $\lambda _0$ such that $L\psi _0=\lambda _0\psi _0$. 
\end{lem}

\begin{proof}
Let $H_1(M)$ be the Sobolev space\footnote{The Sobolev spaces  $H_k\bigl(M^n\bigr)$  for any nonnegative integer 
$k$ can be defined as the completion of smooth functions in the norm  
$$||f||_{(k)}= \sqrt{\Bigl( \sum_{i=0}^k\, || D^{(i)}f ||_2^2 \Bigr)}$$  
where $||\ ||_2$ is the $L^2$ norm, $D^{(0)}f\equiv f$, and  
$D^{(i)}f\equiv \bigl( D_{a_1a_2\dots a_i}fD^{a_1a_2\dots a_i}f 
\bigr) ^{1/2}$.  } of all $L^2(M)$ 
functions whose generalized derivative\footnote{The {\it generalized derivative} with respect to $x^i$ of a locally 
integrable function on an open subset $U\subseteq {\mathbb R}^n$ is the locally 
integrable function $\partial _if$ such that 
$$\int_U\, \phi \partial_if\, d^nx=-\int_U\, f{{\partial \phi}\over {\partial x^i}}\, d^nx$$ 
for all $\phi \in C^\infty(U)$. 
 } is also in $L^2(M)$. $H_1(M)$ 
is actually a Hilbert space with inner product given by 
$(\phi ,\psi )=\int_M\, (D_a\phi D^a\psi + \phi \psi )dV$ . 
Now, define a functional in the following way: 
$$I(\psi )={{\int_M\bigl(D_a\psi D^a\psi + V\psi ^2\bigr)dV}\over{\int_M
\psi^2\, dV}}$$ 
for $\psi$ in $H_1(M)$ and not identically zero. Let ${\cal S}\! =\bigl\{\, 
b\in {\mathbb R}|\ b\le I(\psi)\ {\rm for\ all\ }\psi\in H_1(M)\ {\rm and}\ \psi 
\not\equiv 0\bigr\}$. Since the gradient term is nonnegative and $V$ is 
bounded, $I$ is bounded from below and consequently ${\cal S}$ is nonempty. 
Further, 
${\cal S}$ is bounded above. Since ${\cal S}$ is a nonempty set of real numbers 
 which is bounded above it has a least upper bound, call it $\lambda _0$. 
By definition, it follows that    
$$\lambda _0 = \inf_{ {\psi \in H_1(M)}- \{ 0\}}I(\psi)\ .$$ 
Furthermore, it follows that there is a sequence $\{ \psi _k\}$ in $H_1(M)$ 
such that $I(\psi _k)\rightarrow \lambda _0$ as $k\rightarrow \infty$ and 
$\int\nolimits {\psi _k}^2\, dV=1$ . Clearly, $\{ \psi_k \}$ 
is a bounded subset of 
$H_1(M)$. The embedding $H_1(M)\subset L^2(M)$ is {\it compact}, i.e. every 
bounded set is mapped to a compact one. Therefore, there is a subsequence 
$\{ \psi_{k_i}\}$ of $\{ \psi_k\}$ such that $\psi_{k_{i}}\rightarrow\psi _0$ 
in $L^2(M)$. 
Moreover, $I(\psi _{k_i})\rightarrow\lambda_0$ as $k_i\rightarrow
\infty$ because all subsequences of a convergent sequence converge to the same 
limit. 
Next, it will be shown that $\psi_0\in H_1(M)$ and 
$I(\psi_0)=\lambda_0$. 

Since $\{ \psi_{k_i}\}$ is a bounded subset of $H_1(M)$, $\{ \psi _{k_i}\}$ 
is compact in the weak topology. This is because bounded subsets of Hilbert 
spaces are weakly compact. Hence there is a subsequence $\{ \psi_{k_j}\}$ 
of $\{ \psi_{k_i}\}$ which converges weakly in $H_1(M)$ to $\phi _0$. Now, the 
embedding $H_1(M)\subset L^2(M)$ is continuous so it must also be continuous 
in the weak topology. Therefore, $\{ \psi_{k_j}\}$ converges weakly in $L^2(M)$ 
to $\phi _0$. However, $\{ \psi_{k_i}\}$ converges in $L^2(M)$ to $\psi _0$ by 
the previous arguments. Hence $\psi_{k_i}\rightarrow\psi_0$ weakly in $L^2(M)$.
 This means $\psi _0=\phi _0$ in $L^2(M)$ because $\{ \psi_{k_j}\}$ is a 
subsequence of the weakly convergent sequence $\{ \psi_{k_i}\}$. 
Therefore, $\psi_0$ is in $H_1(M)$. 
From now on, let the sequence $\{ \psi_{k_j}\}$ be denoted by $\{ \psi_m\}$. 
Weak convergence of $\{ \psi_m\}$ means that for all linear functionals, $A$,  
$A(\psi_m)\rightarrow A(\psi_0)$ as $m\rightarrow\infty$.  
In particular, weak convergence and Schwarz's inequality imply 
$(\psi_0, \psi_0)\le \lim_{m\rightarrow\infty}(\psi_m,\psi_0) \le 
\lim_{m\rightarrow\infty}||\psi_m||||\psi_0||$, where $||\ ||$ and $(\ ,\ )$  
are the norm and inner product on $H_1$. Hence, $(\psi_0,\psi_0)\le 
\lim_{m\rightarrow\infty} (\psi_m,\psi_m)$. 
Further, the $L^2(M)$ convergence of $\{ \psi_m\}$ and the fact that $V$ is 
bounded imply $\int\nolimits (V-1){\psi_m}^2\, dV\rightarrow\int\nolimits
(V-1){\psi_0}^2\, dV$. This is proven by using H\"older's inequality\footnote{
 H\"older's inequality is: Given $f\in L^p$ and $g\in L^q$, then 
$$\int |fg|dV\le \Bigl(\int |f|^p\, dV\Bigr)^{1\over p}\Bigl(\int |g|^q\, dV
\Bigr)^{1\over q}$$ for ${1\over p}+{1\over q}=1$ with $1\le p\le \infty$.} 
and the triangle inequality, namely,  

\begin{align}
& \Big|\int\nolimits (V-1){\psi_m}^2\, dV -\int\nolimits
(V-1){\psi_0 }^2\, dV\Big| \le K\int\nolimits
|{\psi_m}^2-{\psi_0}^2|dV  \le   K
||\psi_m-\psi_0||_2 \nonumber \\ & ||\psi_m+\psi_0||_2 \le K||\psi_m -\psi_0||_2
(||\psi_m||_2 + ||\psi_0||_2) \le K||\psi_m
-\psi_0||_2(1+||\psi_0||_2 ) \nonumber
\end{align}

where $|V-1|\le K$ and $||\ ||_2$ is the $L^2(M)$ norm. Choosing $V=2$, this 
argument also implies that $||\psi_0||_2=1$.  
Using the above limits, it follows that 

\begin{align*}
& \lambda_0=\lim_{m\rightarrow\infty}I(\psi_m)
=\lim_{m\rightarrow\infty} \frac{ (\psi_m,\psi_m)+\int\nolimits
(V-1) { \psi_m}^2\, dV } {||\psi_m||_2^2}  \\ & \ge
\frac{(\psi_0,\psi_0)+\int\nolimits (V-1) { \psi_0}^2\, dV} {||\psi_0
||_2^2} =I(\psi_0)\ge \lambda_0\ .
\end{align*}

Therefore, $\psi_0$ is in $H_1(M)$ and an extremum of $I$. 

The smoothness of $\psi_0$ will now be demonstrated. Since the functional $I$ 
attains an extremum, the derivative at $\psi_0$ must vanish in all 
directions, i.e.   
$$D_\phi I(\psi_0)=\int_M(D_a\psi_0D^a\phi + V\psi_0\phi -\lambda_0\psi_0\phi )
dV=0 $$ 
In particular, it vanishes for smooth $\phi$ which means 
$$-D^2\psi_0 + V\psi_0=\lambda_0\psi_0$$ 
in the sense of distributions or in other words $\psi_0$ is a weak solution. 
Now, rewrite the equation as $D^2\psi_0=(
\lambda_0-V)\psi_0$ which means $D^2\psi_0$ is in $L^2(M)$ because $\psi_0$ 
is in $H_1(M)$ and $V$ is smooth. Further, $|D_a\psi_0|$ is in $L^2(M)$ and  
$V$ is smooth so it follows that $|D_aD^2\psi_0|$ is in $L^2(M)$. By using 
induction, one can prove that $D^{2k}\psi_0$ and $|D_aD^{2k}\psi_0|$ are in 
$L^2(M)$. Now, one can prove that the following identity  
$$D_aD_bT_{cd\dots e}-D_bD_aT_{cd\dots e}=R_{abc}^{\ \ \ \ f}T_{fd\dots e} 
+R_{abd}^{\ \ \ \ f}T_{cf\dots e}+\dots +R_{abe}^{\ \ \ \ f}T_{cd\dots f} 
\  ,$$ 
holds. Using this identity and the facts about the derivatives of $\psi_0$ 
, it follows that 
$$\int\nolimits D_aD_b\psi_0D^aD^b\psi_0\, dV=\int\nolimits \bigl(D^2\psi_0
\bigr)^2\, dV-\int\nolimits R_{ab}D^a\psi_0D^b\psi_0\, dV $$ 
and the righthand side is finite. Therefore, $\psi_0\in H_2(M)$. Repeating 
this procedure for $D_aD_bD_c\psi_0$, it follows that 

\begin{align*}
& \int\nolimits D_aD_bD_c\psi_0D^aD^b D^c\psi_0\, dV= \int\nolimits
D_cD^2 \psi_0D^cD^2\psi_0\, dV -\int\nolimits R_{be}D^eD_c\psi_0D^b
D^c\psi_0\, dV - \\ & 2\int\nolimits R_{abce}D^aD^e\psi_0D^bD^c\psi_0\, dV -
\int\nolimits \bigl(D^aR_{abcd}\bigr)D^d\psi_0D^bD^c\psi_0\, dV - \\ &
\int\nolimits R_c^{\, d}D_d\psi_0 R_e^{\, c}D^e\psi_0\, dV-
2\int\nolimits R_c^{\ d}D_d\psi_0D^cD^2\psi_0\, dV
\end{align*}

and again the righthand side is finite. Therefore, $\psi_0\in H_3(M)$. 
By using this bootstrap technique and induction, one can prove that any 
weak solution to the differential equation is in $H_k(M)$ for all $k$. 
However, the Sobolev embedding theorem for compact n-manifolds says that 
$$ H_k(M)\subset C^r(M)$$ 
for $k>{n\over 2} +r$. Therefore, all weak solutions are smooth. 
  
Finally, let $\psi_0$ be any solution, then the above argument implies it is 
smooth. Since $\psi_0$ is smooth, it follows that $|\psi_0|$ is also in 
$H_1(M)$. Furthermore, $I(\psi_0)=I(|\psi_0|)$ so $|\psi_0|$ is also a 
solution. Further, the regularity of weak solutions implies $|\psi_0|$ must be 
smooth. Therefore, $\psi_0$ cannot change sign which 
means $\psi_0$ can be chosen to be nonnegative. In order to prove that 
$\psi_0>0$, let $u= -\psi_0$ and then apply the maximum principle 
theorem \ref{Theorem 10}. 
\end{proof} 
  
The next lemma is used in producing asymptotically flat 3-manifolds from 
closed ones. Although, one usually proves this using operator techniques on 
Hilbert spaces, it will proven via a variational principle below. 

\begin{lem}\label{Lemma 13} If $\lambda _0>0$ for the operator $L=-D^2+V$, then 
$L$ is an isomorphism of $C^\infty (M)$ to itself. 
\end{lem}

\begin{proof}
First, $\lambda_0>0$ and its variational characterization imply  
$$0< \lambda_0 ||\phi ||_2^2\le \int_M\, \bigl(D_a\phi D^a\phi + V\phi^2\bigr)
dV$$  
for all $\phi \in C^\infty (M)$ not identically zero. This means $L$ is one to 
 one. Now, define an inner product on $C^\infty (M)$ by 
$$(\phi ,\psi )_L\equiv \int_M\, \bigl(D_a\phi D^a\psi + V\phi \psi \bigr)dV$$ \
Furthermore, there is a constant $K>0$ such that $K^{-1}||\phi ||\le ||\phi ||_L
\le K||\phi ||$ where $||\ ||$ is the standard norm on $H_1(M)$. Therefore, the 
completion of $C^\infty (M)$ in the norm $||\ ||_L$ is not only a Hilbert space 
but it is $H_1(M)$.  

Next, pick a smooth function $f$ and define the functional 
$F(\psi )\equiv \int f\psi dV$ for all $\psi \in H_1(M)$.  Observe that 
$$|F(\psi )|\le \int_M\, |f\psi |dV\le ||f||_2||\psi ||_2\le C||\psi ||_L $$ 
where $C$ is a constant. Therefore, $F$ is a bounded linear functional on 
$H_1(M)$. The Riesz representation theorem\footnote{Riesz Representation Theorem: Given any bounded linear functional 
$F$ on a Hilbert space ${\cal H}$ there is a uniquely determined vector $f$ 
in ${\cal H}$ such that $F(x)=(x,f)$ for all $x\in {\cal H}$. Furthermore, 
$||F||=||f||$. Note: ${\cal H}$ need not be separable.}  implies there is a unique 
$\phi \in H_1(M)$ such that $(\psi ,\phi)=F(\psi )$ for all $\psi 
\in H_1(M)$. Therefore, 
$$\int_M\, \bigl(D_a\psi D^a\phi + V\psi \phi -f\psi \bigr)dV=0$$ 
for all $\psi$, in other words $\phi$ is a weak solution of $L\phi =f$. 
Since $f$ and $V$ are smooth, the regularity argument of lemma \ref{Lemma 12}
implies $\phi$ is smooth.
\end{proof}

The following Sobolev inequality is needed to control the norm of the scalar 
curvature as the metric is deformed using a family of bump functions. This 
inequality is well known and proven in several different papers \cite{sy}. 
However, the proof will be given here for the sake of completeness. 

\begin{lem}\label{Lemma 14} If $\Sigma $ is an asymptotically flat 3-manifold 
with metric $g_{ab}$ satisfying $g_{ab}=\delta _{ab}+h_{ab}$ outside a compact 
set $C$ where $|h_{ab}|\le {A\over r}$, $|\partial _ch_{ab}|\le {A\over r^2}$, 
and $|\partial _d\partial _ch_{ab}|\le {A\over r^3}$; then 
$$\Bigl(\int_\Sigma |f|^6\, dV\Bigr)^{1\over 6}\le K\Bigl(\int_\Sigma D_afD^af\, dV\Bigr)^{1\over 2}$$ 
for all $f\in C_o^\infty\bigl(\Sigma \bigr)$. Furthermore, if two metrics agree 
 on $C'\supset C$ and are bounded by the same constant $A$, then $K$ is the 
same for both of them. 
\end{lem}

\begin{proof}
The asymptotic behavior of the metric implies that there is 
a constant $\lambda $ depending on $A$ such that 
$\lambda ^{-1}\delta _{ab}\xi^a\xi^b\le g_{ab}\xi^a\xi^b\le \lambda \delta_{ab}\xi^a\xi^b$ on $\Sigma -C$.  
Now, the above Sobolev inequality does hold on ${\mathbb R}^3$ minus a ball with 
the standard metric. This follows from the same inequality on ${\mathbb R}^3$. 
The asymptotic bounds on the metric and the Sobolev inequality for the flat 
metric on ${\mathbb R}^3-B$ imply 

\begin{align*}
& \Bigl(\int_{\Sigma -C}\, |f|^6\, dV\Bigr)^{1\over 6}\le \lambda
 ^{1\over 4}\Bigl(\int_{\Sigma -C}\, |f|^6\, d^3x\Bigr)^{1\over 6}
 \le {\overline K}\lambda ^{1\over 4}\Bigl(\int_{\Sigma -C}\,
 \delta^{ab}D_afD_bf\, dV\Bigr)^{1\over 2} \le \\ & {\overline K}\lambda
 ^{2\over 3}\Bigl(\int_{\Sigma -C}\, D_afD^af\, dV\Bigr)^{1\over 2}
 \end{align*}

where ${\overline K}$ is the Sobolev constant for ${\mathbb R}^3-B$ times the number
 of asymptotic regions. This establishes the Sobolev inequality on the 
asymptotic regions. Now, suppose the inequality fails on $\Sigma $; then there 
is a function $f_n\in C_o^\infty\bigl(\Sigma \bigr)$ for positive integer $n $
such that $\int_\Sigma |f_n|^6\, dV=1$ and $\int_\Sigma \, D_af_nD^af_n\, dV<{1\over n}$.  
Now, apply the Sobolev inequality to the sequence $f_n$ restricted to $\Sigma -C$ 
to obtain $\bigl(\int_{\Sigma -C}\, |f_n|^6\, dV\bigr)^{1\over 6}< {{{\overline K}\lambda ^{3\over 2}}\over n^{1\over 2}}$.  
Taking the limit as $n\rightarrow \infty $ yields $f_n\rightarrow 0$ 
in the $L^6$ norm on $\Sigma -C$. 

One other inequality is needed to complete the proof, namely, the following:  
Given any compact 3-manifold (with or without boundary) $N$ and any smooth 
function on it which satisfies $\int_N\, fdV=0$, then there is a constant 
$\widetilde K$ independent of $f$ such that 
$$\Bigl(\int_N\, |f|^6\, dV\Bigr)^{1\over 6}\le {\widetilde K}\Bigl(\int_N\, D_afD^af\, dV\Bigr)^{1\over 2}\ .$$ 
One can prove this inequality by choosing a covering of $N$ by charts and a 
partition of unity. Next, apply the Sobolev inequality 
on ${\mathbb R}^3$ to each chart and sum the resulting inequalities using the 
partition of unity; for details see \cite{a}. This yields the inequality 
$$\Bigl(\int_N\, |f|^6\, dV\Bigr)^{1\over 6}\le {\widehat K}\Bigl[\Bigl(\int_N\, D_afD^af\, dV\Bigr)^{1\over 2} + \Bigl(\int_N\, |f|^2\, dV\Bigr)^{1\over 2}\Bigr]\ .$$ 
The additional term is due to the gradients of the functions in the partition
of unity. In order to express the last integral in terms of the $L^2$ norm of 
the gradient of $f$, the first nontrivial eigenvalue of the Laplacian on $N$ is 
used. More precisely, let 
$$\lambda_1=\inf_{\phi }{{\int_N\, D_a\phi D^a\phi \, dV}\over {\int_N\, \phi^2\, dV}}\, ,$$ 
subject to the constraint $\int_N \phi dV=0$. Since $\lambda _1$ is a 
minimum of the functional on the righthand side, it follows that 
$$\int_N\, f^2\, dV\le {1\over {\lambda_1}}\int_N\, D_afD^af\, dV\, ,$$ 
for all smooth $f$ which satisfy $\int_N  fdV=0$. Combining this inequality 
with the previous one yields the desired inequality. 
 
Now, let $\beta _n\equiv \int_{C_1}\, f_n\, dV$ where $C\subset C_1\subseteq C'$
 and $C_1$ is compact, then 
$$\Bigl(\int_{C_1}\, |f_n-\beta _n|^6\, dV\Bigr)^{1\over 6}\le {\widetilde K}\Bigl(\int_{C_1}\, D_a{f_n}D^a{f_n}\, dV\Bigr)^{1\over 2}<{{\widetilde K}\over {n^{1\over 2}}}\ ,$$ 
because of the inequality in the previous paragraph. Taking the limit as 
$n\rightarrow\infty$ yields $f_n\rightarrow\beta _n$ in the $L^6$ norm on $C_1$ 
 which means $f_n\rightarrow\beta _n$ in the $L^6$ norm on $\bigl(\Sigma - C\bigr)\bigcap C_1$.  
However, $f_n\rightarrow 0$ on $\Sigma -C$ from our arguments on the asymptotic
 regions, which implies	$\beta _n\rightarrow 0$. Therefore, $f_n\rightarrow 0$ 
in the $L^6$ norm on $\Sigma $, a contradiction to 
$\int_\Sigma \, |f_n|^6\, dV=1$. Therefore, the Sobolev 
inequality must hold for all $f\in C_o^\infty\bigl(\Sigma \bigr)$. 

Finally, if two metrics agree on $C'$ and are asymptotically bounded by the 
same constants, then the Sobolev constant is the same for both of them. This 
is true because the metrics agree on a compact set and must have the same 
Sobolev constant on that region. Next, in the asymptotic regions both metrics 
have the same bounding constants. Moreover, the argument which established 
the inequality on the asymptotic regions implies both metrics have the same 
constant in the asymptotic regions. Therefore, it follows that $K$ can be 
chosen to be the same for both metrics. 
\end{proof}

The final proposition is an existence theorem for conformal metrics 
with $R=0$. This result was proven by several authors \cite{sy,cb}, 
however, the proof is given here for the sake of completeness and the technique 
used here is simpler. The norm $||\ ||_{3\over 2}$ is the norm on 
$L^{3\over 2}$.  

\begin{lem}\label{Lemma 15} If $\Sigma $ is an asymptotically flat 3-manifold and 
$L= -D^2 + V$, then $L\phi =f$ has a unique solution which is smooth and 
$O({1\over r})$ whenever $V$ and $f$ are smooth, $O({1\over r^3})$, and 
$||V_{-}||_{3\over 2} < {1\over K^2}$ where $K$ is a Sobolev constant (from lemma 
\ref{Lemma 14}) and $V_{-}$ is the absolute value of the negative part of $V$. 
\end{lem}

\begin{proof}
Let 
$$||\phi ||_L^2\equiv \int \bigl( D_a\phi D^a\phi + V\phi ^2
\bigr) dV$$ for all $\phi \in C_o^\infty (\Sigma )$. Using H\"older's 
inequality and the fact that $V=V_{+}-V_{-}$ where $V_{+}$ 
and $V_{-}$ are absolute values of the nonnegative and negative 
parts of $V$, respectively, it follows that 
$$\int_\Sigma D_a\phi D^a\phi\, dV - ||V_{-}||_{3\over 2}||\phi ^2||_3 
\le \int_\Sigma \bigl( D_a\phi D^a\phi - V_{-}\phi^2\bigr)dV\le 
\int_\Sigma \bigl( D_a\phi D^a\phi + V\phi^2\bigr)dV \ $$  
where $||\ ||_{{3\over 2}}$ and $||\ ||_3$ are the $L^{{3\over 2}}$ and 
$L^3$ norms respectively. Next, observe that observe that 
$||\phi ||_{2p}^2=||\phi ^2||_{p}$ for $L^{{2p}}$ and 
$L^p$ norms, respectively so

$$\int_\Sigma D_a\phi D^a\phi\, dV - ||V_{-}||_{3\over 2}||\phi ||_6^2 =
\int_\Sigma D_a\phi D^a\phi\, dV - ||V_{-}||_{3\over 2}||\phi ^2||_3 
\le 
\int_\Sigma \bigl( D_a\phi D^a\phi + V\phi^2\bigr)dV \  $$ 

Lemma \ref{Lemma 14} implies that 
$$0<\bigl( {1\over K^2} - ||V_{-}||_{3\over 2}\bigr) ||\phi||_6^2\le 
||\phi ||_L^2 \ .$$  
Therefore, $||\ ||_L$ is a norm and its completion, $H_L$, is contained in 
$L^6(\Sigma )$. Furthermore, $H_L$ is a Hilbert space because $||\phi ||_L^2 
=(\phi ,\phi )_L$ where 
$$(\phi ,\psi )_L \equiv \int_{\Sigma } \bigl( D_a\phi D^a\psi + V\phi \psi 
\bigr) dV\ .$$ 

Next, let $F(\psi )\equiv \int f\psi dV$; then H\"older's inequality 
implies 
$$|F(\psi )|\le \int_{\Sigma } |f\psi |dV\le ||f||_{6\over 5}||\psi ||_6 
\le C||\psi ||_L $$ 
where $C$ is a constant. Hence $F$ is a bounded linear functional on $H_L$. 
The Riesz representation theorem implies there exists a unique 
$\phi \in H_L$ such that $(\psi ,\phi )_L=F(\psi )$ for all $\psi \in H_L$. 
Therefore, $L\phi =f$ weakly. 

Finally, $\phi \in L^6$ which means $\phi$ is locally in $L^6$.\footnote{ A function is locally in $L^p(\Sigma )$ if it is in $L^p(K)$ for 
all compact sets, $K$, contained in $\Sigma $. The space of all such functions 
is denoted by $L_{\rm loc}^p(\Sigma )$ and the topology is given by requiring 
sequences to converge in $L^p(K)$ for all $K$. }
Since $L_{\rm loc}^6(\Sigma )\subset L_{\rm loc}^2(\Sigma )$, it follows that 
$\phi \in L_{\rm loc}^2(\Sigma )$. This means the regularity argument of 
lemma \ref{Lemma 12} applies locally. Hence, $\phi $ is locally smooth on $\Sigma$ 
and therefore smooth everywhere. Furthermore, using the fall off 
conditions on $V$ and $f$, one can show $\phi$ has ${1\over r}$ fall off. 
\end{proof}
 
The compactification theorem will now be proven using a gluing procedure. 
The technique involves deforming the metric to a flat metric in the asymptotic 
regions, and regulating the norm of the scalar curvature so that the operator 
$L=-8D^2 + R $ remains positive. Then, the asymptotic regions are compactified 
by smoothly gluing in a ball onto each of the regions. Next, positivity of the 
operator $L$ is used to construct a Green's function of the operator $L$ on the 
 compactified manifold. Finally, the properties of the Green's function are used
 to find a smooth metric on the compactification with $R>0$.
   
\begin{thm}\label{Theorem 16} Every asymptotically flat 3-manifold $\Sigma $ which 
admits a metric with nonnegative scalar curvature has a smooth compactification 
 $\widetilde \Sigma $ which admits a metric with positive scalar curvature. 
\end{thm}

\begin{proof}
Let $g_{ab}$ be a metric on $\Sigma $ with $R\ge 0$. Now, 
define a one parameter family of metrics ${\bar g}_{ab}(t)\equiv
\alpha _tg_{ab} + (1-\alpha _t)\delta _{ab}$ where $\alpha _t$ is a smooth 
function on $\Sigma $ 
which is equal to one in the interior region of $\Sigma $; and on the asymptotic
 regions has the behavior $\alpha =1$ for $r\le t$, $\alpha =0$ for $r\ge 2t$, 
$|\alpha '|\le {C\over r}$, and $|\alpha ''|\le {C\over r^2}$ where $C$ is a 
constant independent of $t$. The existence of $\alpha $ follows from lemma \ref{Lemma 11}. 
Using the fact that $g_{ab}$ is asymptotically flat and the behavior of $\alpha $,  
it follows that ${\bar g}_{ab}(t)$ is a one parameter family of asymptotically 
flat metrics. The goal is to find a value of $t$ for which the corresponding 
metric is conformal to one which has $R=0$ and then use the conformal factor 
as a Green's function on the compactified manifold. In general, one can never 
have $R\ge 0$ for any metric in the family unless $\Sigma ={\mathbb R}^3$ and 
$g_{ab}$ is flat. 

In order to find the conformal factor, the procedure is to pick $t$ large 
enough, and therefore make $R_{-}$ small enough, so that lemma \ref{Lemma 15} applies with 
$V={1\over 8}R$ and $f=-V$. In order to prove the bound of lemma \ref{Lemma 15} recall 
that the family of metrics has the form ${\bar g}_{ab}(t)= \delta_{ab} 
+ \alpha_th_{ab}$ in the asymptotic regions, where $h_{ab}=O({1\over r})$. 
Further, there is a constant $A$ such that $|\alpha h_{ab}|\le {A\over r}$, 
$|\partial _c\bigl(\alpha h_{ab}\bigr)|\le {A\over r^2}$, and 
$|\partial _d\partial _c\bigl(\alpha h_{ab}\bigr)|\le {A\over r^3}$ for all 
large t. These results combined with the additional observation that $R_{-}=0$ 
unless $t\le r\le 2t$ imply that there are two constants ${\bar A}$ and 
$B$ both independent of $t$ such that 
$$\Bigl(\int_{\Sigma }\, |R_{-}|^{3\over 2}\, dV\Bigr)^{2\over 3}
\le \Bigl(\int_{\Sigma }\, |R|^{3\over 2}\, dV\Bigr)^{2\over 3}
\le {\bar A}\Bigl(
\int_t^{2t}\, {\Bigl({1\over r^3}\Bigr)}^{3\over 2}\, r^2 dr\Bigr)^{2\over 3}
\le {B\over t}\ .$$ 
Now, fix a $t$ large enough so that ${{K^2B}\over t}<8$; denote this metric 
by ${\bar g}_{ab}$. 
For the metric ${\bar g}_{ab}$, lemma \ref{Lemma 15} can be applied and it follows that 
there is a smooth $\psi =O({1\over r})$. Let ${\overline G}\equiv 1+\psi $, 
then ${\overline G}$ satisfies 
$$-8{\overline D}^2\, {\overline G} + {\overline R}\, {\overline G}=0$$ 
on $\Sigma$ with respect to ${\bar g}_{ab}$. Next, it will be shown that 
${\overline G}$ is positive. 
 
Define a family of operators $L_{\lambda }= -8{\overline D}^2 + {\lambda }
{\overline R} $ for $0\le \lambda \le 1$. Lemma \ref{Lemma 15} still applies and it 
follows that a family of solutions ${\overline G}_{\lambda }=1+
{\psi }_{\lambda }$ exists. Let $m_{\lambda }$ denote the minimum of 
${\overline G}_{\lambda }$; then $m_{\lambda }$ is a continuous function 
of $\lambda $. Suppose that for some $\lambda $, $m_{\lambda }\le 0$; then 
continuity implies there is some smaller value of $\lambda $ for which 
$m_{\lambda }=0$. The maximum principle then implies ${\overline G}_{\lambda }=
0$ for all $x\in \Sigma$. This contradicts the fact that ${\overline G}_{\lambda
 } \rightarrow 1$ as $r\rightarrow \infty $. Therefore, ${\overline G}_{\lambda 
} >0$ for $0\le \lambda \le 1$. In particular, it holds for $\lambda =1$. 

The manifold $\Sigma $ will now be compactified. 
Recall in the asymptotic regions ${\bar g}_{ab}=\delta _{ab} + \alpha h_{ab}$ 
where $\alpha $ is as above, in particular $\alpha =0$ for $r>2t$. Now, define 
a new metric ${\tilde g}_{ab}\equiv \phi ^4 {\bar g}_{ab}$ where 
$\phi \equiv \gamma +{(1-\gamma)\over r}$ and $\gamma $ is a smooth decreasing 
function which is one for $r<3t$ and zero for $r>4t$. Given $r>4t$ the metric 
becomes ${\tilde g}_{ab}= {1\over r^4}\delta _{ab}$. Because of the form of 
this metric, the point at infinity $i_k$ on each asymptotic region can be 
smoothly added. Hence ${\widetilde \Sigma }=\Sigma \bigcup \{i_1,i_2,\dots ,i_n
\}$ is smooth and ${\tilde g}_{ab}$ is a smooth metric on $\widetilde \Sigma $. 

Next, a positive Green's function is constructed for the operator 
$L=-8{\widetilde D}^2 + {\widetilde R}$ with respect to the metric 
${\tilde g}_{ab}$. Let 
${\overline G}>0$ be the conformal factor found above for the metric 
${\bar g}_{ab}$ on $\Sigma $. The function $\overline G$ satisfies the equation 
$$-8{\overline D}^2\, {\overline G} + {\overline R}\, {\overline G}=0$$ 
on $\Sigma $ with respect to the metric ${\bar g}_{ab}$. It follows that 
${\widetilde G}\equiv \phi ^{-1}{\overline G}$ satisfies the equation 
$-8{\widetilde D}^2{\widetilde G}+{\widetilde R}{\widetilde G}=0$, 
on ${\widetilde \Sigma}-\{i_1,i_2,\dots ,i_n\}$ 
with respect to the metric ${\tilde g}_{ab}$. Now, the goal is to prove that 
$\widetilde G$ is the Green's function of the operator $L=-8{\widetilde D}^2+{\widetilde R}$ 
on ${\widetilde \Sigma }$ with respect to the metric ${\tilde g}_{ab}$. 
For $r>2t$, it follows that 
$${\overline G}_k=1+{{M_k}\over {2r}}+ \sum_{l=1}^\infty \sum_{m=-l}^l\, {{A_{lm}(k)Y_{lm}(\theta ,\phi )}\over {r^{l+1}}}$$ 
where $k$ denotes which asymptotic region, the $Y_{lm}$'s are spherical harmonics, 
$M_k$ is the mass measured at $i_k$, and the $A_{lm}(k)$'s are higher moments 
measured at $i_k$. Thus for $r>4t$, 
$${\widetilde G}_k= {1\over {\bar r}} +{M_k\over 2} + \sum_{l=1}^\infty \sum_{m=-l}^l\, A_{lm}(k)Y_{lm}(\theta ,\phi )\, {{\bar r}^l}$$ 
where ${\bar r}\equiv {1\over r}$. Since the point at infinity $i_k$ for each 
asymptotic region has radial coordinate ${\bar r}=0$, it follows that 
$L{\widetilde G}_k=8\delta _{i_k}$ on the compactification of each asymptotic 
region. It is important to note that the delta function at the point ${i_k}$ is denoted by $\delta _{i_k}$ and it 
is linear functional on smooth test functions of the compactified manifold. It has the
well known feature that $\langle \delta _{i_k},\phi \rangle=c \phi  ({i_k}) $, for any test 
function $\phi $ where $c$ is a positive constant determined by the choice of normalization of the delta function.

Hence, $L{\widetilde G}= 8\sum_{k=1}^n\delta _{i_k}$ on ${\widetilde \Sigma }$. 
Furthermore, ${\widetilde G}>0$ because both $\phi $ and ${\overline G}$ are 
positive. Therefore, ${\widetilde G}$ is a positive Green's function for $L$ 
on $\widetilde \Sigma $. 

Finally, the above Green's function is used to construct a metric with 
positive scalar curvature on $\widetilde \Sigma $. Let $\psi _0$ denote the 
ground state and $\lambda _0$ the corresponding eigenvalue of the operator $L$. 
By lemma \ref{Lemma 12}, a smooth positive ground state $\psi _0$ always exits. Now, using 
the fact that $\widetilde G$ is the Green's function, it follows that 
$$8\sum_{k=1}^n\, \psi _0(i_k)=\sum_{k=1}^n\, \frac 1c \langle \delta _{i_k},\psi _0
\rangle =\frac 1c \langle L{\widetilde G},\psi _0\rangle =\frac 1c \langle {\widetilde G}, 
L\psi _0\rangle =\frac{\lambda _0}{c}\langle {\widetilde G}, \psi _0\rangle\ .$$ 
This combined with the positivity of both $\psi _0$ and $\widetilde G$ imply 
$\lambda _0>0$. Finally, define a new smooth metric on $\widetilde \Sigma $ by 
${\hat g}_{ab}\equiv {\psi _0}^4{\tilde g}_{ab}$. The scalar curvature of 
${\hat g}_{ab}$ is $${\widehat R}={\psi _0}^{-5}\bigl(-8{\widetilde D}_a{\widetilde D}^a\psi _0 + {\widetilde R}\psi _0\bigr)=\lambda _0\psi _0^{-4}>0\ .$$ 
Therefore, $\widetilde \Sigma $ has a metric with positive scalar curvature. 
\end{proof} 

Conversely, given any closed 3-manifold with positive scalar curvature one can 
remove points to obtain an asymptotically flat 3-manifold with metric having 
$R=0$. The converse of the compactification was first suggested by 
R. Geroch \cite{rg}. Combining the above compactification theorem with 
Geroch's result yields the following corollary:

\begin{cor}\label{Corollary 17} An asymptotically flat 3-manifold $\Sigma $ has a 
metric with $R=0$ if and only if it has a compactification $\widetilde \Sigma $ 
which has a metric with ${\widetilde R}>0$. 
\end{cor}

\begin{proof}
The compactification is an immediate consequence of theorem \ref{Theorem 16}.
 Given a closed 3-manifold $\widetilde \Sigma $ with metric having 
${\widetilde R}>0$. Let $L$ be the operator $-8{\widetilde D}^2 + {\widetilde R}
$. Take any finite set of points $\{ i_1,i_2,\dots ,i_n\}$ in $\widetilde 
\Sigma$ and solve the equation $L\phi = {\bar \delta }$ where ${\bar \delta}
\equiv 8\sum_{k=1}^n\delta_{i_k}$. Since ${\widetilde R}>0$, $\lambda_0 >0$ and 
lemma \ref{Lemma 13} implies $L$ is an isomorphism on $C^\infty\bigl({\widetilde \Sigma}
\bigr)$. Let $\{ f_n\}$ be a sequence of smooth functions which converge to 
${\bar \delta}$ in the sense of distributions. Since $L$ is an isomorphism 
on smooth functions, there is a sequence of smooth functions $\{ \phi_n \}$ 
such that $L\phi_n = f_n$. It follows that 
$$\langle \phi_n -\phi_m , L\psi \rangle = \langle L(\phi_n -\phi_m) ,\psi 
\rangle = \langle f_n -f_m,\psi \rangle $$  
for all smooth $\psi$. Furthermore, the righthand side goes to zero as 
$m,n\rightarrow \infty$ because $\{ f_n\} $ converges. This implies $\langle 
\phi_n -\phi_m , {\tilde \psi}\rangle \rightarrow 0$ for all smooth functions 
${\tilde \psi}$ because $L$ is an isomorphism. Hence $\{ \phi_n \}$ is a Cauchy 
 sequence in the space of distributions. Moreover, completeness of the space of 
distributions implies $\phi_n \rightarrow \phi$. Therefore, $L\phi =
{\bar \delta }$. 

Since $L$ has smooth coefficients and $L\phi =0$ on ${\widetilde \Sigma}- 
\{ i_1,i_2,\dots ,i_n\}$, a standard bootstrap argument implies $\phi $ must 
also be smooth on ${\widetilde \Sigma}-\{i_1,i_2,\dots ,i_n\}$. Using normal 
coordinates about each of the points $i_k$, one can show that $\phi$ has the 
correct asymptotic behavior on ${\widetilde \Sigma}-\{i_1,i_2,\dots ,i_n\}$. 
Finally, an argument using the maximum principle implies that $\phi >0$.
\end{proof}
     
Now, the compactification theorem and the classification theorem are combined 
to prove that most asymptotically flat globally hyperbolic spacetimes do not admit maximal slices. 
This result is generic in that it only depends on the spatial topology of the 
spacetime. 

\begin{thm}\label{Theorem 18} Most asymptotically flat globally hyperbolic spacetimes do not admit a 
maximal slice. 
\end{thm}

\begin{proof}
The existence of a maximal slice implies $\Sigma $ admits a 
metric with $R\ge 0$. The compactification theorem implies the closed manifold 
$\widetilde \Sigma $ admits a metric with positive scalar curvature. This 
combined with theorem \ref{Theorem 8} implies most $\Sigma $ never admit a metric with 
$R\ge 0$. Theorem \ref{Theorem 7} implies all $\Sigma $ are allowed as the spatial topology 
of spacetimes. Therefore, most asymptotically flat globally hyperbolic spacetimes never admit 
a maximal slice. 
\end{proof}

A more explicit description of the spatial topology of globally hyperbolic spacetimes that 
\underbar{do} \underbar{admit} a maximal slice can be given by 
decomposing the orientable closed 3-manifolds in terms of three basic types of prime 
orientable factors. They are closed 3-manifolds with finite fundamental group, 
the $K(\pi ,1)$'s ,and the handle $S^2\times S^1$. 
By theorem \ref{Theorem 8}, the only prime factors which can possibly admit a 
metric with $R>0$ are the ones with finite fundamental group and the handle. 
The handle does have have a metric with positive scalar curvature, one such 
metric is the product metric. The prime 3-manifolds with finite fundamental group
are the spherical spaces. Furthermore, the only orientable prime factors 
admitting metrics with $R>0$ are these spherical spaces and the 
handle $S^2\times S^1$.

If $M^3$ admits a metric with $R\ge 0$, then theorem \ref{Theorem 8} implies $M^3$ either 
has a metric with positive scalar curvature in which case we know its form from our previous 
discussion, or it is flat but there are only six orientable closed flat spaces. Thus, the 
form of any closed orientable 3-manifold with $R\ge 0$ is completely known if 
these conjectures are assumed. This means that, given any spatially closed globally hyperbolic 
spacetime which admits a maximal slice the spatial topology is known, more 
precisely it is the connected sum of spherical spaces and handles, or it is one of six 
flat spaces. Of course, if the globally hyperbolic spacetime is spatial nonorientable, then the 
double cover is of the above form. In the asymptotically flat case, these 
results imply that the spatial topology is the connected sum of spherical 
spaces and handles minus a finite number of points. 

\section{ CONCLUSIONS} 

Although the generic situation is that asymptotically flat globally hyperbolic spacetimes do not 
admit maximal slices, the question still remains what are both necessary and 
sufficient conditions for the maximal slices to exist. It is important to 
realize that just because there are no topological obstructions to maximal 
slices does not mean that maximal slices exist. This is best illustrated 
with the following example: Let the spatial topology be that of a 3-torus $T^3$.
 If one takes ${\mathbb R}\times T^3$ with the metric $ds^2=-dt^2+dx^2+dy^2+dz^2$, 
i.e. Minkowski space with points identified, then the $t={\rm constant}$ slices 
 yield a foliation  by maximal slices. On the other hand, if one takes 
the same topology but with a different metric, namely, $ds^2=-dt^2+Ct^{4\over 3}\bigl(dx^2+dy^2+dz^2\bigr)$ 
where $C>0$ is a constant, i.e. a dust filled Robertson-Walker spacetime with 
points identified, then this spacetime has no maximal slice because $T^3$ 
expands as $t$ increases. This is an example of a spatially closed spacetime 
but with a little bit of work it is conceivable that such examples can be 
constructed in the asymptotically flat case also. 

Likewise, it has been shown that CMC hypersurfaces which are not maximal can 
exist regardless of the spatial topology but this does not assure that they 
always exist. This only means that there are no topological obstructions to 
finding such hypersurfaces. In fact, examples of spatially closed vacuum spacetimes 
which do not have any CMC hypersurface given by \cite{ew,Chrusciel:2004qe}. 

\section{ ACKNOWLEDGEMENTS}

The author is grateful to NSERC for financial support.  In addition, the author 
would like to thank the Perimeter Institute for its hospitality during some of the 
work on this project. The author would like to thank D. M. Eardley, J. L. Friedman, 
G. T. Horowitz, J. B. Hartle, and R. P. Woodard for useful conversations. Especially, 
the author would like to thank J. Brannlund and K. Schleich for very useful conversations
and for comments on this manuscript.

\end{document}